\documentclass[12pt]{article}
\usepackage{amsmath}
\usepackage{graphicx,psfrag,epsf}
\usepackage{enumerate}
\usepackage{natbib}
\usepackage{url} 
\usepackage{amsthm,amsmath,natbib}
\usepackage{amsfonts}
\usepackage{graphicx}
\usepackage{psfrag,epsf}
\usepackage{enumerate}
\usepackage{natbib}
\usepackage{url} 
\usepackage{color}
\usepackage{bbm}
\usepackage{bm}
\usepackage{todonotes}
\usepackage[para,online,flushleft]{threeparttable}
\usepackage{amsthm,amssymb}
\usepackage{rotating}
\usepackage{pdflscape}
\usepackage{float}
\usepackage[normalem]{ulem}

\newcommand{\E}{\mathbb E}
\newcommand{\I}{\mathbbm 1}

\newtheorem{theorem}{Theorem}

\newcommand\independent{\protect\mathpalette{\protect\independenT}{\perp}}
\def\independenT#1#2{\mathrel{\rlap{$#1#2$}\mkern2mu{#1#2}}}

\newcommand{\blind}{1}

\begin{document}

\def\spacingset#1{\renewcommand{\baselinestretch}%
{#1}\small\normalsize} \spacingset{1}

\date{March 2019}
\if1\blind
{
  \title{\bf Optimal balancing of time-dependent confounders for marginal structural models}
   \author{Nathan Kallus\\
    School of Operations Research and Information Engineering and \\
Cornell Tech, Cornell University, New York, New York 10044\\
\\
    Michele Santacatterina\thanks{
    Corresponding author. This material is based upon work supported by the National Science Foundation under Grants Nos. 1656996 and 1740822.}\hspace{.2cm}\\
    TRIPODS Center for Data Science for Improved Decision Making \\
    and Cornell Tech, Cornell University, New York, New York, 10044}
  \maketitle
} \fi

\if0\blind
{
  \bigskip
  \bigskip
  \bigskip
  \begin{center}
    {\LARGE\bf Title}
\end{center}
  \medskip
} \fi
\newpage
\bigskip
\begin{abstract}
Marginal structural models (MSMs) estimate the causal effect of a time-varying treatment
in the presence of time-dependent confounding via weighted regression.
The standard approach of using inverse probability of treatment weighting (IPTW)
can lead to high-variance estimates due to extreme weights and be sensitive to model misspecification.
Various methods have been proposed to partially address this, including truncation and stabilized-IPTW to temper extreme weights and covariate balancing propensity score (CBPS) to address treatment model misspecification.
In this paper, we present Kernel Optimal Weighting (KOW), a convex-optimization-based approach that finds weights for fitting the MSM that optimally balance time-dependent confounders while simultaneously penalizing extreme weights, directly addressing the above limitations. 
We further extend KOW to control for informative censoring. We evaluate the performance of KOW in a simulation study, comparing it with IPTW, stabilized-IPTW, and CBPS. We demonstrate the use of KOW in studying the effect of treatment initiation on time-to-death among people living with human immunodeficiency virus and the effect of negative advertising on elections in the United States.
\end{abstract}

\noindent%
\vfill

\newpage
\spacingset{1.45} 
\section{Introduction} \label{sec:introduction}
Marginal structural models (MSMs) offer a successful way to estimate the causal effect of a time-varying treatment on an outcome of interest from longitudinal data in observational studies \citep{robins2000marginal,robins2000marginalb}. For example, they have been used to estimate the optimal timing of HIV treatment initiation \citep{hiv2011initiate}, to evaluate the effect of hormone therapy on cardiovascular outcomes \citep{hernan2008observational}, and to evaluate the impact of negative advertising on election outcomes \citep{blackwell2013framework}.
The increasing popularity of MSMs among applied researchers derives from their ability to control for time-dependent confounders, which are confounders that are affected by previous treatments and affect future ones. In particular, as shown by \cite{robins2000marginalb} and \cite{blackwell2013framework}, standard methods, such as regression or matching, fail to control for time-dependent confounding, introducing post-treatment bias. In contrast, MSMs consistently estimate the causal effect of a time-varying treatment via inverse probability of treatment weighting (IPTW), which controls for time-dependent confounding by weighting each subject under study by the inverse of their probability of being treated given covariates, {\textit{i.e.}, the propensity score \citep{rosenbaum1983central},} mimicking a sequential randomized trial. In other words, IPTW creates a hypothetical pseudo-population where time-dependent confounders are balanced over time.  

\textcolor{black}{Despite their wide range of applications, the usage of these methods in observational studies may be jeopardized by their considerable dependence on}
 positivity. This assumption requires that, at each time period, the probability of being assigned to the treatment, conditional on the history of treatment and confounders, is not 0 or 1 \citep{robins2000marginal}. \textcolor{black}{Even if positivity holds theoretically, when propensities are close to 0 or 1, it can be \textit{practically} violated}.
Practical positivity violations lead to extreme and unstable weights, which in turn yield very low precision and misleading inferences \citep{kang2007demystifying,robins1995analysis,scharfstein1999adjusting}. In addition, MSMs using IPTW are highly sensitive to \textit{misspecification} of the treatment assignment model, which can 
lead to biased estimates 
\citep{kang2007demystifying,lefebvre2008impact,cole2008constructing}.

Various statistical methods have been proposed in an attempt to overcome these challenges.  To deal with extreme weights, several authors \citep{cole2008constructing,xiao2013comparison} have suggested  truncation, whereby outlying weights are replaced with less extreme ones.  \cite{santacatterina2019optimal} proposed to use shrinkage instead of truncation as a more direct way to control the bias-variance trade-off.
\cite{robins2000marginalb} recommended the use of stabilized-IPTW (sIPTW) where inverse probability weights are normalized by the marginal probability of treatment. To control for misspecification of the treatment assignment model, \cite{imai2015robust} proposed to use the covariate balance propensity score (CBPS), which instead of plugging in a logistic regression estimate of propensity into IPTW finds the logistic model that balances covariates via the generalized method of moments. The method tries to balance the first moment of each covariate even if a logistic model is misspecified \citep{imai2014covariate}. 



\textcolor{black}{In this paper, we present and apply Kernel Optimal Weighting (KOW), which provides weights for fitting an MSM that optimally balance time-dependent confounders while controlling for precision. Specifically, 
by solving a quadratic optimization problem over weights, the proposed method directly minimizes \textit{imbalance}, defined as the sum of discrepancies between the weighted observed data and the counterfactual of interest over all treatment regimes, while penalizing extreme weights.} 

This extends the kernel optimal matching method of \citet{kallus2016generalized} and \citet{kallus2018more} to the longitudinal setting and to dealing with time-dependent confounders, where, similarly to regression and matching, it cannot be applied without introducing post-treatment bias.



{The proposed method has several attractive characteristics. First, by optimally balancing time-dependent confounders while penalizing extreme weights, it leads to better accuracy, precision, and total error. 
In particular, in the simulation study presented in Section \ref{simu}, we show that the mean squared error (MSE) of the estimated effect of a time-varying treatment obtained by using KOW is lower than that obtained by using IPTW, sIPTW, and CBPS in all considered simulated scenarios.  Second, differently from \cite{imai2015robust}, where the number of covariate balancing conditions grows exponentially in the number of time periods, KOW only needs to minimize a number of {discrepancies} that grows linearly in the number of time periods. 
This feature leads to a lower computational time of KOW compared with CBPS when the total number of time periods increases, as shown in our simulation study in Section \ref{simu_comp_time} and in our study on the effect of negative advertising on election outcomes in Section \ref{caseblack}. Third, by optimally balancing covariates, KOW mitigates the effects of possible misspecification of the treatment model. In Section \ref{simu}, we show that KOW is more robust to model misspecification compared with the other methods. Fourth, KOW can balance non-additive covariate relationships by using kernels, which generalize the structure of conditional expectation functions, and does not restrict weights to follow a fixed logistic (or other parametric) form. 
In Section \ref{simu}, we show how KOW compares favorably with the aforementioned methods in all nonlinear scenarios, and in Section \ref{caseblack} we use KOW to balance non-additive covariate relationships estimating the effect of negative advertising on election outcomes.
Fifth, KOW can be easily generalized to other settings, such as informative censoring. We do just that in Section \ref{censor} and, in Section \ref{casehiv}, we use this extension to study the effect of human immunodeficiency virus (HIV) treatment on time to death among people living with HIV. Finally, KOW can be solved by using off-the-shelf solvers for quadratic optimization. }

 
In the next section, we briefly introduce the literature of MSMs (Section \ref{sec:rew_msm}). In Section \ref{kow} we develop and define KOW. We then discuss some practical guidelines on the use of KOW (Section \ref{guidelines}). In Section \ref{simu} we report the results of a simulation study aimed at comparing KOW with IPTW, sIPTW, and CBPS. In Section \ref{censor}, we extend KOW to control for informative censoring. We then present two empirical applications of KOW in medicine and political science (Section \ref{empirics}). We offer some concluding remarks in Section \ref{conclusions}.

\section{Marginal structural models for longitudinal data} \label{sec:rew_msm}

In this section, we briefly review MSMs \citep{robins2000marginal,robins2000marginalb}. Suppose we have a simple random sample with replacement of size $n$ from a population. 
For each unit $i=1, \ldots, n$ and time period $t=1, \ldots, T$, we denote the binary time-varying treatment variable by $A_{it}$, with $A_{it}=0$ meaning not being treated at time $t$ and $A_{it}=1$ being treated at time $t$, and time-dependent confounders $X_{it}$. We denote by $\overline{A}_{it}=\lbrace A_{i1}, \ldots, A_{it} \rbrace$ the treatment history up to time $t$ and by $\overline{X}_{it}=\lbrace X_{i1}, \ldots, X_{it} \rbrace$ the history of confounders up to time $t$. $X_{i1}$ represents the time-invariant confounders, \textit{i.e.}, confounders that do not depend on past treatments. 
{We denote by $\overline{a}_t$ and $ \overline{x}_t$ possible realizations of the treatment history $\overline{A}_{it}$ and the confounder history $\overline{X}_{it}$, respectively.
We use $\mathbbm{1}[\cdot]$ to denote the indicator so that $\mathbbm1\left[\overline A_{it}=\overline a_t\right]$ is the variable that is 1 if $\overline A_{it}=\overline a_t$ and 0 otherwise.
To streamline notation, we will refer to $\overline{A}_{iT}$ as $\overline{A}_{i}$, $\overline{a}_{T}$ as $\overline{a}$, $\overline{X}_{iT}$ as $\overline{X}_{i}$, and to $\overline{x}_{T}$ as $\overline{x}$.
} 
For each unit $i=1, \ldots, n$, we denote by $Y_i$ the outcome variable observed at the end of the study. Using the potential outcome framework \citep{imbens2015causal}, we denote by $Y_i(\overline{a})$ the potential outcome we would see if we were to apply the treatment regime $\overline{a}\in\mathcal A$ to the $i^\text{th}$ unit,
where $\mathcal A=\{0,1\}^T$ is the space of treatment regimes. 
Throughout, we drop the subscripts $i$ on these variables to refer to a generic unit.

We impose the assumptions of consistency, non-interference, positivity and sequential ignorability \citep{imbens2015causal,hernan2010causal}. Consistency and non-interference \citep[also known as SUTVA;][]{rubin1980randomization} can be encapsulated in that the potential outcomes are well-defined and the observed outcome corresponds to the potential outcome of the treatment regime applied to that unit, \textit{i.e.,} 
$Y=Y(\overline{A})$.
As previously introduced, positivity states that, for each time $t=1, \ldots, T$, the probability of being treated at time $t$ conditioned on the treatment history up to time $t-1$ and the confounder history up to time $t$, is not 0 or 1, \textit{i.e.},
\begin{equation}
\label{positivity}
	0 < \mathbb P(A_{t} = 1 \mid \overline{A}_{t-1}, \overline{X}_{t} ) < 1 \quad \forall \ t \in \lbrace 1, \ldots, T \rbrace,
\end{equation}
%
Sequential ignorability states that the potential outcome $Y(\overline{a})$ is independent of treatment assignment at time $t$, given the treatment history up to time $t-1$ and the confounder history up to time $t$.  Formally, sequential ignorability is defined as
\begin{equation}
\label{igno}
	Y(\overline{a}) \independent A_{t} \mid \overline{A}_{t-1}, \overline{X}_{t} \quad \forall \ t \in \lbrace 1, \ldots, T \rbrace.
\end{equation}

 An MSM is a model for the marginal causal effect of a time-varying treatment regime on the mean of $Y$, that is,
\begin{equation}
\label{MSM}
	\mathbbm{E} \left[   Y(\overline{a})  \right] = g(\overline{a},\bm{\beta}),
\end{equation}
where $g(\overline{a},\bm{\beta})$ is some known function class parametrized by $\bm\beta$. For example, a commonly used MSM is based on additive effects with a common coefficient: $g(\overline{a},\bm{\beta})=\beta_1 + \beta_2 \sum_{t=1}^Ta_{t}$, where the parameter $\beta_2$ is the causal parameter of interest. 
{Usually, $\bm\beta$ is computed by a weighted regression of the outcome on the treatment regime alone using weighted least squares (WLS), \textit{i.e.,} $\min_{\bm\beta}\sum_{i=1}^nW_i(Y_i-
g(\overline A_{i},\bm\beta)
)^2$, and Wald confidence intervals are constructed using robust (sandwich) standard errors \citep{freedman2006so,robins2000marginal,hernan2001marginal}.
In order to consistently estimate $\bm\beta$, the weights $W_{1:n}=(W_1, \dots, W_n)$, must account for the non-randomness of the treatment assignment mechanism, \textit{i.e.,} 
the confounding.
\citet{robins2000marginal} showed that the set of inverse probability weights and stabilized inverse probability weights achieve this objective. These weights are defined as follows,}
\begin{equation}
\label{sipweights}
	\begin{aligned}
			 W_i^{\text{IPTW}}=w(\overline{A}_{i}, \overline{X}_{i}),\quad w(\overline{a}, \overline{x}) 	&= \prod_{t=1}^T \frac{h_t(\overline a_t)}{\mathbb P(A_{t}=a_t \mid \overline{A}_{t-1}=\overline a_{t-1}, \overline{X}_{t}=\overline x_t )},
	\end{aligned}
\end{equation}


%
%
\noindent
where $h_t(\overline{a}_{t})$ is a known function of treatment history. The set of inverse probability weights is obtained by setting $h_t(\overline{a}_{t})=1$, while the set of stabilized inverse probability weights is obtained by setting $h_t(\overline{a}_{t})=\mathbb P(A_{t} = a_{t}  \mid \overline{A}_{t-1} = \overline{a}_{t-1})$. 
To estimate weights of the form of eq.~\eqref{sipweights},
one first estimates the conditional probability models using
either parametric methods such as logistic regression or other machine learning methods \citep{karim2017application,gruber2015ensemble,karim2017estimating} and then these estimates are plugged in directly into eq.~\eqref{sipweights} to derive weights, which are then plugged into the WLS. Stabilized weights seek to attenuate the variability of inverse probability weights by normalizing them by the marginal probability of treatment. Since the additional factor is a function of treatment regime alone, it does not affect the consistency of the WLS if the MSM is well specified. Both sets of weights, however, rely on plugging in an estimate of a probability into the denominator, meaning that when the true probability is even modestly close to 0, any small error in estimating it can translate to very large errors in estimating the weights and to estimated weights that are extremely variable. 
{Furthermore, both sets of weights rely on the correct specification of the conditional probability models used to estimate the weights in eq.~\eqref{sipweights}.
}

To overcome this issue, 
{\citet{imai2015robust} proposed to estimate weights of the form of eq.~\eqref{sipweights} that improve balance of confounders by generalizing the covariate balancing propensity score (CBPS) methodology.}  Instead of plugging in probability estimates based on logistic regression, CBPS uses the generalized method of moments to find the logistic regression model that if plugged in would lead to weights, $W_i^{\text{CBPS}}$,  that approximately solve a subset of the moment conditions that the true inverse probability weights, eq.~\eqref{sipweights}, satisfy. 

{
Differently than IPTW, sIPTW and CBPS, in the next Section, we characterize imbalance as the discrepancies in observed average outcomes due to confounding, consider their worst case values, and use quadratic optimization to obtain weights that directly optimize the balance of time-invariant and time-dependent confounders over all possible weights while controlling precision.}

\section{{Kernel Optimal Weighting}} 
\label{kow} 

In this Section we present a convex-optimization-based approach that obtains weights that minimize the imbalance due to time-dependent confounding (\textit{i.e.,} maximize balance thereof) while controlling precision. Toward that end, in Section \ref{imbalance}, we provide a definition of imbalance. Specifically, we define imbalance as the sum of discrepancies between the weighted \textit{observed} data and the \textit{unobserved} counterfactual of interest over all treatment regimes. Since this imbalance depends on unknown functions, in Section \ref{nswci} we consider the worst case imbalance, which guards against all possible realizations of the unknown functions. We also show that the worst case imbalance has the attractive characteristic that the number of discrepancies considered grows \textit{linearly} in the number of time periods and not \textit{exponentially} like the number of treatment regimes.  We finally show how to minimize this quantity while controlling precision using kernels, reproducing kernel Hilbert space (RKHS) and off-the-shelf solvers for quadratic optimization  (Sections \ref{min}-\ref{kernelqp}). 



\subsection{Defining imbalance}
\label{imbalance}

Consider any population weights $W=w(\overline{A}, \overline{X})$, where $w(\cdot)$ is a function that depends on the treatment and confounder histories. 
In this Section, we will
show that, under 
consistency and assumptions \eqref{positivity}--\eqref{igno}, we can decompose the difference between the weighted average outcome among the $\overline a$-treated units, $\mathbbm{E} [ W \mathbbm{1}[\overline A=\overline a] Y ]$, and the average potential outcome of $\overline a$, $\mathbbm{E} [   Y(\overline a)  ]$, into a sum over time points $t$ of discrepancies 
involving the values of treatment and confounder histories up to time $t$.

To build intuition we start by explaining this decomposition in the case of two time periods $T=2$.
Assuming consistency and assumptions \eqref{positivity}--\eqref{igno}, for each $\overline a=(a_1,a_2)\in\mathcal A$, we can decompose the weighted average outcome among the $\overline a$-treated units as follows:
\begin{align}
\label{deco}
\E[W\I[\overline A=\overline a]Y]&=
\E[W\I[A_1=a_1]\I[A_2=a_2]\E[Y(\overline a)\mid A_1,A_2,X_1,X_2]]
\\\notag&=
\E[W\I[A_1=a_1]\I[A_2=a_2]\E[Y(\overline a)\mid A_1,X_1,X_2]]
\\\notag&=
\E[W\I[A_1=a_1]\E[Y(\overline a)\mid A_1,X_1,X_2]]
+\delta^{(2)}_{a_2}(W,g_{\overline a}^{(2)})
\\\notag&=
\E[W\I[A_1=a_1]\E[Y(\overline a)\mid X_1]]
+\delta^{(2)}_{a_2}(W,g_{\overline a}^{(2)})
\\\notag&=
\E[Y(\overline a)]
+\delta^{(1)}_{a_1}(W,g_{\overline a}^{(1)})+\delta^{(2)}_{a_2}(W,g_{\overline a}^{(2)})
\\\notag&= \E[Y(\overline a)] + \sum_{t=1}^{2}\delta^{(t)}_{a_t}(W,g_{\overline a}^{(t)}),
\end{align}	
where the first equality follows from iterated expectation, the second from sequential ignorability, the fourth from iterated expectation and sequential ignorability and the third and fifth from the following definitions, which exactly capture the difference between the two sides of the third and fifth equalities,
\begin{align}	
 \delta^{(2)}_{a_2}(W,h^{(2)}) &= \mathbbm{E} \left[ W \mathbbm{1}[A_2=a_2] h^{(2)}(A_1,X_1,X_2) \right] - \mathbbm{E} \left[W h^{(2)}(A_1,X_1,X_2) \right] \\\notag
    g_{\overline a}^{(2)}(A_1,X_1,X_2) &=  \mathbbm{1}[A_1=a_1]\mathbbm{E} \left[ Y(\overline a) \mid A_1,X_1,X_2 \right]\\\notag
		\delta^{(1)}_{a_1}(W,h^{(1)}) &= \mathbbm{E} \left[ W  \mathbbm{1}[A_1=a_1] h^{(1)}(X_1) \right] - \mathbbm{E} \left[   h^{(1)}(X_1) \right] \\\notag
    g_{\overline a}^{(1)}(X_1) &=  \mathbbm{E} \left[ Y(\overline a) \mid X_1 \right].
\end{align}
Note our use of $h^{(t)}$ as a 
generic dummy function and 
$g_{\overline a}^{(t)}$ as a \textit{specific} function that depends on the particular (unknown) distribution of $\overline X_t,\overline A_{t-1},Y(\overline a)$.

This gives a definition of discrepancy, $\delta^{(t)}_{a_t}(W,h^{(t)})$, where the subscript $a_t\in\{0,1\}$ refers to the treatment assigned at time $t$,  $W=w(\overline{A}, \overline{X})$ is a population weight, and $h^{(t)}$ is a given function of interest of the treatment and confounder history up to $t$, $\overline A_{t-1},\overline X_t$. The function $g_{\overline a}^{(t)}$ is one such function.
In particular, for every $a_1\in\{0,1\}$,
the quantity $\delta^{(1)}_{a_1}(W,h^{(1)})$ is the discrepancy between the $h^{(1)}$-moments of the baseline confounder distribution in the weighted $a_1$-treated population and of the distribution in the whole population. Similarly, for every $a_2\in\{0,1\}$, $\delta^{(2)}_{a_2}(W,h^{(2)})$ is a discrepancy in the $h^{(2)}$-moment of treatment and confounder histories at the start of time step $2$.
What we have shown above is how these discrepancies directly relate to the difference between weighted averages of observed outcomes and true averages of unknown counterfactuals of interest.
Specifically, we have shown that when we measure these discrepancies with respect to the specific function $g_{\overline a}^{(t)}$, then their sum gives that difference.



We can extend this decomposition to general horizons $T\geq1$.
Let us define the same discrepancies for any time $t\geq3$ as
\begin{equation*}
	\begin{aligned}
	\delta^{(t)}_{a_t}(W,h^{(t)}) &= \mathbbm{E} \left[ W \mathbbm{1}[A_t=a_t] h^{(t)}(\overline A_{t-1},\overline X_t) \right] - \mathbbm{E} \left[ 
    W
    h^{(t)}(\overline A_{t-1},\overline X_t) \right],\\
    g_{\overline a}^{(t)}(\overline A_{t-1},\overline X_t) &=  \mathbbm{1}[\overline A_{t-1}=\overline a_{t-1}]\mathbbm{E} \left[ Y(\overline a) \mid \overline A_{t-1},\overline X_t \right].
    \end{aligned}
\end{equation*}
The following result gives the general decomposition of the difference between weighted average of observed outcomes and true average of counterfactuals as the sum of $T$ discrepancies, one for every time step:
\begin{theorem}
\label{thm1}
Under assumptions \eqref{positivity}--\eqref{igno}, for each $\overline a\in\mathcal A=\{0,1\}^T$, 
\begin{equation*}
 \mathbbm{E} \left[ W \mathbbm{1}[\overline A=\overline a] Y \right]
- \mathbbm{E} \left[ Y(\overline a) \right]
=
\sum_{t=1}^T\delta^{(t)}_{a_t}(W,g_{\overline a}^{(t)}). 
\end{equation*}
\end{theorem}


Based on the results of Theorem \ref{thm1}, it is clear that if we want the difference between average counterfactual outcomes and average weighted factual outcomes to be small for all treatment regimes $\overline a$ then we should seek weights $W$ that make 
$$\overline\delta_{\overline a}(W,\overline g_{\overline a})=\sum_{t=1}^T\delta^{(t)}_{a_t}(W,g_{\overline a}^{(t)})
$$
small for all $\overline a$,
where we write $\overline h=(h^{(1)},\dots,h^{(T)})$ for any set of $T$ functions. 


%
%

The empirical counterparts to
$\delta^{(t)}_{a_t}(W,h^{(t)})$
are the sample moment discrepancies for a given set of sample weights $W_{1:n}$:
\begin{equation}
\label{deltase}
	\begin{aligned}
		\hat\delta^{(t)}_{a_t}(W_{1:n},h^{(t)}) &= \frac1n\sum_{i=1}^n (W_i\mathbbm{1}[A_{it}=a_t]-
    W_i
    ) h^{(t)}(\overline A_{i,t-1},\overline X_{it}),\quad \forall t\geq2, 
    \\
		\hat\delta^{(1)}_{a_1}(W_{1:n},h^{(1)}) &= \frac1n\sum_{i=1}^n W_i  \mathbbm{1}[A_{i1}=a_1] h^{(1)}(X_{i1}) - \frac1n\sum_{i=1}^n  h^{(1)}(X_{i1}),  \\
		\hat{\overline\delta}_{\overline a}(W_{1:n},\overline h)&=\sum_{t=1}^T\hat\delta^{(t)}_{a_t}(W_{1:n},h^{(t)}).
	\end{aligned}
\end{equation}
Thus, we will seek samples weights $W_{1:n}$ that make $\hat{\overline\delta}_{\overline a}(W_{1:n},\overline g_{\overline a})$ small for all treatment regimes $\overline a$.
Toward that end, 
for \textit{any} set of given functions 
$(\overline h_{\overline a})_{\overline a\in\mathcal A}$,
%
%
%
%
we define \textit{imbalance} of a set of weights $W_{1:n}$ as the average squared discrepancy over treatment regimes:
\begin{equation}
\label{biasb2}
	\begin{aligned}
\text{IMB}(W_{1:n};(\overline h_{\overline a})_{\overline a\in \mathcal A}) &=  \frac1{\left|\mathcal A\right|}\sum_{\overline a\in\mathcal A}\hat{\overline\delta}^2_{\overline a}(W_{1:n},\overline h_{\overline a}).
	\end{aligned}
\end{equation}
The particular imbalance of interest is given when we consider $\overline h_{\overline a}=\overline g_{\overline a}$.
One way to control this imbalance, $\text{IMB}(W_{1:n};(\overline g_{\overline a})_{\overline a\in \mathcal A})$, and consequently control the empirical discrepancies of interest, $\hat{\overline\delta}_{\overline a}(W_{1:n},\overline g_{\overline a})$, is by using inverse probability weights. If known, these weights make this quantity a sample average of mean-zero variables and thus close to zero for large $n$. However, the difficulties are that (a) even mild practical violations of positivity can lead to large variance of each of these terms and (b) we need to correctly estimate the sequential propensities.

Differently, we will seek to find weights that directly \emph{minimize} imbalance. There are two main challenges in this task. The first challenge is that the imbalance of interest depends on some unknown functions $\overline g_{\overline a}$. The second is that the number of treatment regimes grows exponentially in the number of time periods. In the next Section we show how the proposed methodology overcomes these two challenges.  

\subsection{Worst case imbalance}
\label{nswci}

To overcome the fact that we do not actually know the functions $\overline g_{\overline a}$ on which imbalance $\text{IMB}(W_{1:n};(\overline g_{\overline a})_{\overline a\in \mathcal A})$  depends, we will guard against all possible realizations of the unknown functions. Specifically, since $\hat{\overline\delta}_{\overline a}(W_{1:n},\overline g_{\overline a})$ scales linearly with $\overline g_{\overline a}$, we will consider its magnitude relative to that of $\overline g_{\overline a}$. We therefore need to define a magnitude.
In particular, let us define
$$
\|\overline h\|=
\sqrt{\|h^{(1)}\|^2_{(1)}+\cdots+\|h^{(T)}\|^2_{(T)}},
$$
where 
$\|\cdot\|^2_{(t)}$ are some given extended seminorms on functions from the space of time-dependent confounders and treatment histories up to time $t$ to the space of outcomes. Compared to a norm, an extended seminorm may also assign the values of $0$ and $\infty$ to nonzero elements but must still satisfy triangle inequality and absolute homogeneity. We will discuss specific choices of such seminorms $\|\cdot\|^2_{(t)}$ in Section \ref{kernelqp}.
%

Given these, we can define the \textit{worst case discrepancies}, $$\Delta^{(t)}_{a_t}(W_{1:n})=\sup_{h^{(t)}} \frac{\hat\delta^{(t)}_{a_t}(W_{1:n},h^{(t)})}{\|h^{(t)}\|_{(t)}} =
\sup_{\|h^{(t)}\|_{(t)}\leq1}\hat\delta^{(t)}_{a_1}(W_{1:n},h^{(t)}).$$
Note that $\Delta^{(t)}_{a_t}(W_{1:n})$ depends \textit{only} on the treatment at time $t$, $a_t$, and \textit{not} the whole treatment regime, $\overline a$.

Then the \textit{worst case imbalance} is given by
\begin{equation}
\label{nswcimb}
	\begin{aligned}
\mathcal{B}^2(W_{1:n})
&=
\sup_{
\|\overline h_{\overline a}\|\leq1\;\forall\overline a\in \mathcal A
}
\text{IMB}(W_{1:n};(\overline h_{\overline a}^{(t)})_{\overline a\in \mathcal A}) 
\\&=
\sup_{\overline h_{\overline a},\;\overline a\in \mathcal A}
\frac1{\left|\mathcal A\right|}\sum_{\overline a\in\mathcal A}\frac{\hat{\overline\delta}^2_{\overline a}(W_{1:n},\overline h_{\overline a})}{\|\overline h_{\overline a}\|^2} 
\\&
=\frac12\sum_{t=1}^T(\Delta^{(t)}_0(W_{1:n})^2+\Delta^{(t)}_1(W_{1:n})^2).
	\end{aligned}
\end{equation}
What is important to note is that this shows that the discrepancies of interest are essentially the same regardless of the particular treatment regime trajectory $\overline a$. 
That is, to control the discrepancies for \textit{all} trajectories $\overline a$ for \textit{all} possible realizations of $\overline g_{\overline a}$, at any time point $t$, we are only concerned with the discrepancies of histories $\overline A_{t-1},\overline X_t$ for those units treated at time $t$, $A_t=1$, and for those not, $A_t=0$.
So, while 
the number of treatment regimes grows \textit{exponentially} in the number of periods, we need only to keep track of and minimize a number of discrepancies growing \textit{linearly} in the number of periods $T$. By eliminating each of these linearly-many imbalances, any time-dependent confounding would necessarily be removed, as shown by
Theorem~\ref{thm1}.
In Section \ref{simu_comp_time}, we show how this feature also translates to favorable computational time when dealing with many time periods.

\subsection{Minimizing imbalance while controlling precision}
\label{min}

We can obtain minimal imbalance by minimizing $\mathcal{B}^2(W)$.
However, to control for extreme weights we propose to regularize the weight variables $W_{1:n}$. We therefore wish to find weights that minimizes $\mathcal{B}^2(W_{1:n})$ plus a penalty for deviations from uniform weighting. Formally, we want to solve 
\begin{equation}\label{cmse}
	\begin{aligned}
		\underset{W_{1:n} \in \mathcal{W}}{\min} \quad
     \mathcal{B}^2(W_{1:n}) + \lambda \|W_{1:n}-e\|_2^2,
	\end{aligned}
\end{equation}

\noindent
where $e$ is the vector of ones and $\mathcal{W}= \lbrace W_{1:n}:W_i \geq 0\;\forall i \rbrace$ is the space of nonnegative weights $W_{1:n}$. 
The squared distance of the weights from uniform weights here serves as a convex surrogate for the variance of the resulting MSM (assuming homoskedasticity or bounded residual variances) and $\lambda$ in eq.~\eqref{cmse} can be interpreted as a penalization parameter that controls the trade off between imbalance and precision.
When $\lambda$ is equal to zero, the obtained weights provide minimal imbalance. When $\lambda \rightarrow \infty$, the weights become uniformly distributed leading to an ordinary least squares estimator for the MSM. 

In the next section, we discuss a specific choice of the norm that specified the worst case discrepancies $\Delta_{a_t}^{(t)}(W_{1:n})$, presented in Section \ref{nswci}. Specifically, we show that by choosing an RKHS to specify the norm, we can express the optimization problem in eq.~\eqref{cmse} as a convex-quadratic function in $W_{1:n}$, which can be easily solved by using off-the-shelf solvers for quadratic optimization.


\subsection{RKHS and quadratic optimization to optimally balance time-dependent confounders}
\label{kernelqp}
An RKHS is a Hilbert space of functions which is associated a kernel (the reproducing kernel). Specifically, any positive semi-definite kernel $\mathcal K:\mathcal Z\times\mathcal Z\to\mathbb R$ on a ground space $\mathcal Z$ defines a Hilbert space given by (the unique completion of) the span of all functions $\mathcal K(z,\cdot)$ for 
$z\in\mathcal Z$, endowed with the inner product $\left<\mathcal K(z,\cdot),\mathcal K(z',\cdot)\right>=\mathcal K(z,z')$.
 Kernels are widely used in machine learning to generalize the structure of conditional expectation functions with many applications in statistics \citep{scholkopf2002learning,berlinet2011reproducing,kallus2016generalized,kallus2018optimal}. Commonly used kernels are the polynomial, Gaussian, and Mat\'ern kernels \citep{scholkopf2002learning}. 

The following theorem shows that if $\|\cdot\|_{(t)}$, the norm that specified the worst case discrepancies $\Delta_{a_t}^{(t)}(W_{1:n})$, is an RKHS norm given by the kernel $\mathcal K_t$, then we can express it as a convex-quadratic function in 
$W_{1:n}$.

\begin{theorem}
\label{thm2}
Define the matrix $K_t\in\mathbb R^{n\times n}$ as $$K_{tij}=\mathcal K_t((\overline A_{i,t-1},\overline X_{it}),(\overline A_{j,t-1},\overline X_{jt}))$$ and note that it is positive semidefinite by definition. Then,
if the norm $\|\cdot\|_{(t)}$ is the RKHS norm given by the kernel $\mathcal K_t$,
the squared worst case discrepancies are
\begin{equation*}
\begin{aligned}
    \Delta^{(1)}_{a_1}(W_{1:n})^2 &= \frac1{n^2}W_{1:n}^TI^{(1)}_{a_1}K_1I^{(1)}_{a_1}W_{1:n}-2e^TK_1I^{(1)}_{a_1}W_{1:n}+e^TK_1e,\\
\Delta^{(t)}_{a_t}(W_{1:n})^2 &= 
\frac1{n^2}W_{1:n}^T(I-I^{(t)}_{a_t}) K_t (I-I^{(t)}_{a_t}) W_{1:n},
\end{aligned}
\end{equation*}
\noindent
where $I$ is the identity matrix and $I^{(t)}_{a_t}$ is the diagonal matrix with $\mathbb I[A_{it}=a_t]$ in its $i^\text{th}$ diagonal entry.
\end{theorem}


Based on Theorem \ref{thm2}, we can now express the worst case imbalance, $\mathcal{B}^2(W_{1:n})$, defined in eq.~\eqref{nswcimb}, as a convex-quadratic function. Specifically, let $K_t^\circ={I^{(t)}_{0}K_tI^{(t)}_{0}+I^{(t)}_{1}K_tI^{(t)}_{1}}$, which is given by setting every entry $i,j$ of $K_t$ to 0 whenever $A_{it}\neq A_{jt}$, and $K^\circ = \sum_{t=1}^TK_t^\circ$.
We then get that
\begin{equation}
	\begin{aligned}
\mathcal{B}^2(W_{1:n})&=\frac1{2}\sum_{t=1}^T(\Delta^{(t)}_0(W_{1:n})^2+\Delta^{(t)}_1(W_{1:n})^2)
\\&
=
\frac1{n^2}\biggl(\frac12W_{1:n}^TK^\circ W_{1:n}-e^TK_1W_{1:n}+e^TK_1e\biggr).
	\end{aligned}
\end{equation}

Finally, to obtain weights that optimally balance covariates to control for time-dependent confounding while controlling precision we solve the quadratic optimization problem, 
\begin{equation}
\label{QP}
	\begin{aligned}
\underset{W_{1:n} \in \mathcal{W}}{\min} \quad & 
\frac12W_{1:n}^TK_{\lambda}^\circ W_{1:n}-e^TK_{\lambda}W_{1:n}
	\end{aligned}
\end{equation}
where $K_{\lambda}^\circ = K^\circ+2\lambda I$, $K_{\lambda}=K_1+2\lambda I$. We call this proposed methodology and the result of eq.~\eqref{QP}, Kernel Optimal Weighting (KOW).

\section{Practical guidelines}
\label{guidelines}

Solutions to the quadratic optimization problem (\ref{QP}) depend on several factors. First, they depend on the choice of the kernel and its hyperparameters. There are some existing practical guidelines on these 
choices \citep{scholkopf2002learning,rasmussen2006gaussian}, on which we 
rely as explained below.
 Second, they depend on the penalization parameter $\lambda$. Finally, solutions to eq.~\eqref{QP} depend on the chosen set of lagged covariates to include in each kernel. In this section, we introduce some practical guidelines on how to apply KOW in consideration of these factors. 


For each $t$, the unknown function $g_{\overline{a}}^{(t)}(\overline{A}_{t-1},\overline{X}_t)$ has two distinct inputs: the treatment history and the confounder history. To reflect this structure, we suggest to specify the kernel $\mathcal K_t$ as a \emph{product kernel}, \textit{i.e.},\break $\mathcal K_t((\overline a_{t-1},\overline x_t),(\overline a'_{t-1},\overline x'_t))=\mathcal K_t^{(1)}(\overline a_{t-1},\overline a'_{t-1})\mathcal K_t^{(2)}(\overline x_{t},\overline x'_{t})$ given a treatment history kernel $\mathcal K_t^{(1)}$ and a confounder history kernel $\mathcal K_t^{(2)}$. This simplifies the process of specifying the kernels. We further suggest that for the treatment history to use a linear kernel involving $\ell$ lagged treatments, $\mathcal K_t^{(1)}(\overline a_{t-1},\overline a'_{t-1})=\sum_{s=\max(1,t-\ell)}^{t-1}a_sa'_s$, and for the confounder history to use a polynomial kernel involving the time-invariant confounders and $\ell$ lagged time-dependent confounders, $\mathcal K_t^{(d)}(\overline x_{t},\overline x'_{t})=(1+\theta x_1^Tx'_1+\theta\sum_{s=\max(2,t-\ell+1)}^{t}x_t^Tx'_t)^d$, where $\theta>0$ and $d\in\mathbb N$ are hyperparameters. We discuss the choice of the number of lags and the hyperparameters below. 
In our simulation study in Section \ref{simu}, we show that the MSE of the 
MSM-estimated effect using KOW with a product of linear kernel and a quadratic kernel ($d=2$) outperforms estimates using weights obtained by IPTW, sIPTW and CBPS in all considered simulated scenarios. 
We again use this choice of kernels in our empirical applications of KOW to real datasets in Section \ref{empirics}.
Many other choices of kernel are also possible and may be more appropriate in a particular application, but we suggest the above combination as a generic and successful recipe.

When using kernels, preprocessing the data is an important step. In particular, normalization is employed to avoid unit dependence and covariates with high variance dominating those with smaller ones. Consequently, we suggest, beforehand, to scale the covariates related to the treatment and confounder histories to have mean 0 and variance 1.  

To tune the kernels' hyperparameters and the penalization parameter $\lambda$, we follow \cite{kallus2016generalized} and use the empirical Bayes approach of marginal likelihood \citep{rasmussen2006gaussian}. 
We postulate a Gaussian process prior $g^{(t)} \sim \mathcal{GP}(c_t\mathbf 1,\mathcal{K}_{t}(\theta))$, where $c_t\mathbf 1$ is a constant function and $\mathcal{K}_{t}(\theta_t)$ is a kernel that depends on some set of hyperparameters $\theta_t$. For each $t$, we then maximize the marginal likelihood of seeing the data $Y \sim \mathcal{N}(g^{(t)}(\overline X_t,\overline A_{t-1}), \lambda_t)$ over $\theta_t,\lambda_t,c_t$ and let $\lambda=\sum_{t=1}^T\lambda_t$. It would be more correct to consider the marginal likelihood of observing the partial means of outcomes, but we find that this much simpler approach suffices for learning the right representation of the data ($\theta_t$) and the right penalization parameter ($\lambda$) and it enables the use of existing packages such as \textsf{GPML} \citep{rasmussen2010gaussian}. We demonstrate this 
in the simulations presented in Section \ref{simu}, and in particular in Figures \ref{fig1b} and \ref{fig2b}
we see that this approach leads to a value of the penalization parameter that is near that which minimizes the resulting MSE of the MSM over possible parameters.


Another practical concern is how many lagged covariates to include in each of the kernels $\mathcal{K}_{t}$. When deriving inverse probability weights, it is common to model the denominator in eq.~\eqref{sipweights} by fitting a pooled logistic model \citep{d1990relation} including only the time-invariant confounders, $X_{1}$, the time-dependent confounders at time $t$, $X_{t}$, and the one-time lagged treatment history, $A_{t-1}$, rather than the entire histories, \textit{i.e.}, logit $\mathbb P(A_{t}=a_t \mid \overline{A}_{t-1}=\overline a_{t-1}, \overline{X}_{t}=\overline x_t )=\alpha_t + \beta_1 A_{t-1} + \beta_2 X_{1} + \beta_3 X_{t}$, \citep{hernan2001marginal,hernan2002estimating}. This can be understood as a certain Markovian assumption about the data generating process which simplifies the modeling when $T$ is large. The same can be done in the case of KOW, where we may assume that $g_{\overline{a}}^{(t)}$ is only a function of the one-time lagged treatment, the time-dependent counfounders at time $t$, and the time-invariant confounders, \textit{i.e.}, $g_{\overline{a}}^{(t)}(\overline A_{t-1},\overline X_t)=g_{\overline{a}}^{(t)}(A_{t-1},X_1,X_t)$, and correspondingly let the kernel $K_t$ only depend on $A_{t-1}$, $X_1$, and $X_t$.
More generally, we can consider including any amount of lagged variables, as represented by the parameter $\ell$ in the above specification of the linear and polynomial kernels. In Section \ref{caseblack}, we consider an empirical setting where $T$ is small and specify the kernels using the whole treatment and confounders histories ($\ell=T$). However, in Section \ref{casehiv} we consider a setting where $T$ is large and, following previous approaches studying the same dataset using IPTW with a logistic model of only the one-time lags \citep{hernan2000marginal,hernan2001marginal,hernan2002estimating}, we keep only the baseline and one-time-lagged data in each kernel specification ($\ell=1$).


Certain datasets, such as the one we study in Section \ref{casehiv}, have repeated observations of outcomes at each time $t=1, \ldots, T$. Thus, for each subject, we have $T$ observations to be used to fit the MSM. Correspondingly, each observation should be weighted appropriately. This can be seen as $T$ instances of the weighting problem. For sIPTW, this boils down to restricting the products in the numerator and denominator of eq.~\eqref{sipweights} to be only up to $t$ for each $t=1, \ldots, T$. Similarly, in the case of KOW, we propose to solve eq.~\eqref{QP} for each value of $t=1, \ldots, T$, producing $n \times T$ weights, one for each of the outcome observations, to be used in fitting the MSM. This is demonstrated in Section \ref{casehiv}. 

In the case of a single, final observation of outcome, normalizing the weights, whether IPTW or KOW, does not affect the fitted MSM as it amounts to multiplying the least-squares loss by a constant factor. 
But in the repeated observation setting described above, normalizing each set of weights for each time period separately can help. 
Correspondingly, we can add a constraint to eq.~\eqref{QP} that the mean of the weights must equal one for each time period separately, which we 
demonstrate in
Section \ref{casehiv}.

\section{Simulations}
\label{simu}
In this section, we show the results of a simulation study aimed at comparing the bias and MSE of estimating the cumulative effect of a time-varying treatment on a continuous outcome by using an MSM with weights obtained by each of KOW, IPTW, sIPTW, and CBPS.

\subsection{Setup}
\label{simu_setup}

We considered two different simulated scenarios with $T=3$ time periods, (1) linear, where the treatment was modeled linearly, and (2) nonlinear, where it was modeled quadratically. In both scenarios, we modeled the outcomes nonlinearly so as not to favor our method unfairly.
We tuned the kernel's hyperparameters and the penalization parameter as presented in Section \ref{guidelines} and computed bias and MSE over 1{,}000 replications for each of varying sample sizes, $n=100, \ldots, 1{,}000$. In addition, to study the impact of the penalization parameter $\lambda$ on bias and MSE, in both scenarios, we fixed the sample size to $n=500$ and considered a grid of 25 values for $\lambda$.

{For the linear scenario, we drew the data from the following model: }
\begin{equation*}
    \begin{aligned}
\label{simu_l}
\notag Y_i &= -1.91 + 0.8 \textstyle \sum_{t=1}^T A_{i,t} + 0.5   \textstyle \sum_{k=1}^3 Z_{i,k} +  0.05 \textstyle \sum_{{k\neq m}} Z_{i,k}Z_{i,m} + N(0, 5),  
\end{aligned}
\end{equation*}
where $Z_{i,k} = \sum_{t=1}^T X_{i,t,k}$, $A_{i,t} \sim \text{binom}(\pi_{i,t}^\text{(L)})$, $X_{i,t,k} \sim N(X_{i,t-1,k} + 0.1,1)$, $k=1,2,3$, and
\begin{equation*}
\begin{aligned}
\pi_{i,t}^\text{(L)} &= (1 + \exp(-(0.5 + 0.5 A_{i,t-1} + 0.05X_{i,t,1}+ 0.08X_{i,t,2} -0.03X_{i,t,3} \\\notag &+  0.2 A_{i,t-1}\textstyle \sum_{k=1}^3 X_{i,t,k})))^{-1} .
\end{aligned}
\end{equation*}

For the nonlinear scenario, we drew the data from following model:
\begin{equation*}
\begin{aligned}
\label{simu_nl}
\notag
Y_i &= -21.46 + 0.8 \textstyle \sum_{t=1}^T A_{i,t} + 0.5  \textstyle \sum_{k=1}^3 Z_{i,k} +  0.1 (\textstyle \sum_{{k\neq m}} Z_{i,k}Z_{i,m}) + N(0, 5),
\end{aligned}
\end{equation*}
\noindent
where $Z_{i,k} = \sum_{t=1}^T X_{i,t,k}^2$, $A_{i,t} \sim \text{binom}(\pi_{i,t}^\text{(NL)})$, $X_{i,t,k} \sim N(X_{i,t-1,k} + 0.1,1)$, $k=1,2,3$ and 
\begin{equation*}
\begin{aligned}
\pi_{i,t}^\text{(NL)} &= (1 + \exp(- (0.5 + 0.5 A_{i,t-1} +
0.05X_{i,t,1} + 0.08X_{i,t,2} -0.03X_{i,t,3} \\\notag &+
0.025X^2_{i,t,1} + 0.04X^2_{i,t,2} -0.015X^2_{i,t,3}
+   0.3 \textstyle \sum_{k \neq m} X_{i,t,k} X_{i,t,m} \\\notag &+  0.2 A_{i,t-1} \textstyle \sum_{k=1}^3 X_{i,t,k} +   0.1 A_{i,t-1} \textstyle \sum_{k=1}^3 X^2_{i,t,k} \\\notag &+ 0.05 A_{i,t-1} \textstyle \sum_{k \neq m} X_{i,t,k} X_{i,t,m})))^{-1}.
\end{aligned}
\end{equation*} 
The intercepts $-1.91$ and $-21.46$ are chosen so the marginal mean of $Y_i$ is 0.
%
%

In each scenario and for each replication, we computed two sets of KOW weights. We obtain the first by using the product of two linear kernels ($\mathcal{K}_1$), one for the treatment history and one for the confounder history, and the second by using the product of a linear kernel for the treatment history and a quadratic kernel for the confounder history ($\mathcal{K}_2$). As presented in Section \ref{guidelines}, we rescaled the variables before inputting them to the kernel and, for each replication, we tuned $\lambda$ and the kernels' hyperparameters by using Gaussian-process marginal likelihood.
%
%
We also computed two sets of IPTW and sIPTW weights. We obtained the first by fitting for each $t=1,2,3$ a logistic regression model for the treatment $\overline A_{i,t}$ conditioned on $\overline A_{i,t-1},\overline X_{i,t}$ and their interactions, which is well-specified for $\pi_{i,t}^\text{(L)}$ (we term this the linear specification) and the second by adding all quadratic confounder terms and their interactions with $\overline A_{i,t-1}$ which is well-specified for $\pi_{i,t}^\text{(NL)}$ (we term this the non-linear specification). 
%
The numerator of sIPTW in either case was obtained by 
fitting
a logistic regression 
on the treatment history alone. We obtain the final set of IPTW and sIPTW weights by multiplying the weights over $t$ as shown in eq. \eqref{sipweights}.
Finally, we computed two sets of weights using CBPS: 
one using the covariates as they are (linear CBPS) 
and one by augmenting the
covariates with all quadratic monomials (non-linear CBPS).
We used the full (non-approximate) version of CBPS.

We computed the causal parameter of interest by using WLS, regressing the outcome on the cumulative treatment and using weights computed by each of the methods. Specifically, in the linear scenario, we computed weights using (1) $\mathcal{K}_1$ for KOW, the linear specification for IPTW and sIPTW, and linear CBPS, which we refer to as the \textit{correct} case, and (2) $\mathcal{K}_2$ for KOW, the nonlinear specification for IPTW and sIPTW, and the nonlinear CBPS, which we refer to as the \textit{overspecified} case. In the nonlinear scenario, we again used each of the above, but refer to the first as the \textit{misspecified} case and the second as the \textit{correct} case. 
We highlight that these terms may only reflect the model specification for IPTW and sIPTW, as CBPS does not require a particular 
specification and
the function $g_{\overline a}^{(t)}$ 
need not necessarily be
in the RKHS that either kernel specify.



We used \textsf{Gurobi} and its \textsf{R} interface \citep{gurobi} to solve eq.~\eqref{cmse} and optimize the KOW weights, the \textsf{MatLab} package \textsf{GPML} \citep{rasmussen2010gaussian} to perform the marginal likelihood estimation of hyperparameters, {the \textsf{R} package \textsf{R.matlab} to call \textsf{MatLab} from within \textsf{R}}, the \textsf{R} command \textsf{glm} to fit treatment models for IPTW and sIPTW, the \textsf{R} package \textsf{CBMSM} for CBPS, and the \textsf{R} command \textsf{lm} to fit the MSM.

\subsection{Results}

In this section we discuss the results obtained in the simulation study across sample sizes and across values of the penalization parameter, $\lambda$.  In summary, the proposed KOW outperformed IPTW, sIPTW and CBPS with respect to MSE across all sample sizes and simulation scenarios. An important result is that, in the misspecified case, KOW showed a lower bias and MSE than that of IPTW, sIPTW and CBPS across all considered sample sizes. 

\subsubsection{Across sample sizes}
\label{simu_ss}

Figure \ref{fig1} shows bias and MSE of the estimated time-varying treatment effect using KOW (solid), 
IPTW (dashed), sIPTW (dotted), and CBPS (dashed-dotted) when increasing the sample size  from $n=100$ 
to $n=1{,}000$. In the linear-correct scenario, IPTW had a lower bias compared with sIPTW, CBPS and KOW in small samples (top-left panel of Figure \ref{fig1}). 
However, for larger samples, KOW had a smaller bias compared with IPTW, sIPTW and CBPS. KOW outperformed IPTW, sIPTW and CBPS in terms of MSE across samples sizes (top-right panel of Figure \ref{fig1}). KOW outperformed the other methods with regards of MSE (bottom-right panel of Figure \ref{fig1}) across all sample sizes, in the linear-overspecified scenario.  KOW and sIPTW performed similarly with respect to bias in the nonlinear-misspecified scenario (top-left panel of Figure \ref{fig2}), while KOW outperformed IPTW, sIPTW and CBPS with respect to MSE in all sample sizes (top-right panel of Figure \ref{fig2}). KOW, IPTW and sIPTW had similar bias in the nonlinear-correct scenario (bottom-left panel of Figure \ref{fig2}), with KOW outperforming the other methods, with respect of MSE, across all sample sizes (bottom-right panel of Figure \ref{fig2}). In summary, the MSE obtained by using KOW was lower than that of IPTW, sIPTW and CBPS across all considered sample sizes.
As the next section shows, the larger biases in some of the cases are driven by the choice of penalization parameter $\lambda$. Here we choose $\lambda$ with an eye toward minimizing MSE. A smaller $\lambda$, it is shown next, can lead to KOW having \emph{both} smaller bias and MSE than other methods, but the total benefit in MSE is smaller.

\begin{figure}[H] 
\begin{center}
\includegraphics[scale=.6]{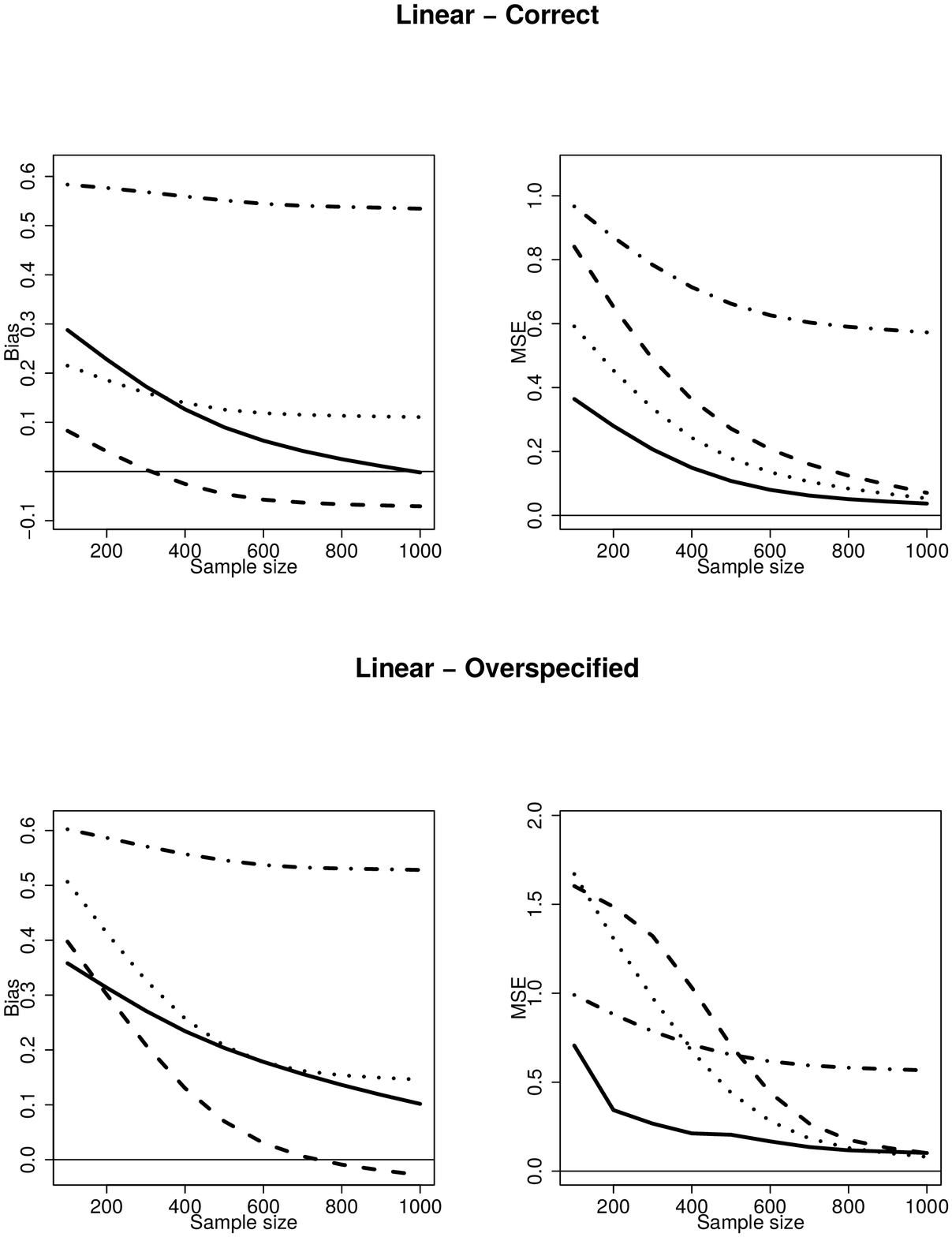}
\end{center}
\caption{\footnotesize Bias and MSE of the estimated time-varying treatment effect using KOW (solid), IPTW (dashed), sIPTW (dotted) and CBPS (dashed-dotted) when increasing the sample size from $n=100$ to $n=1{,}000$ in the linear-correct scenario (top panels) and in the linear-overspecified scenario (bottom panels).
\label{fig1} }
\end{figure}

\begin{figure}[H] 
\begin{center}
\includegraphics[scale=.6]{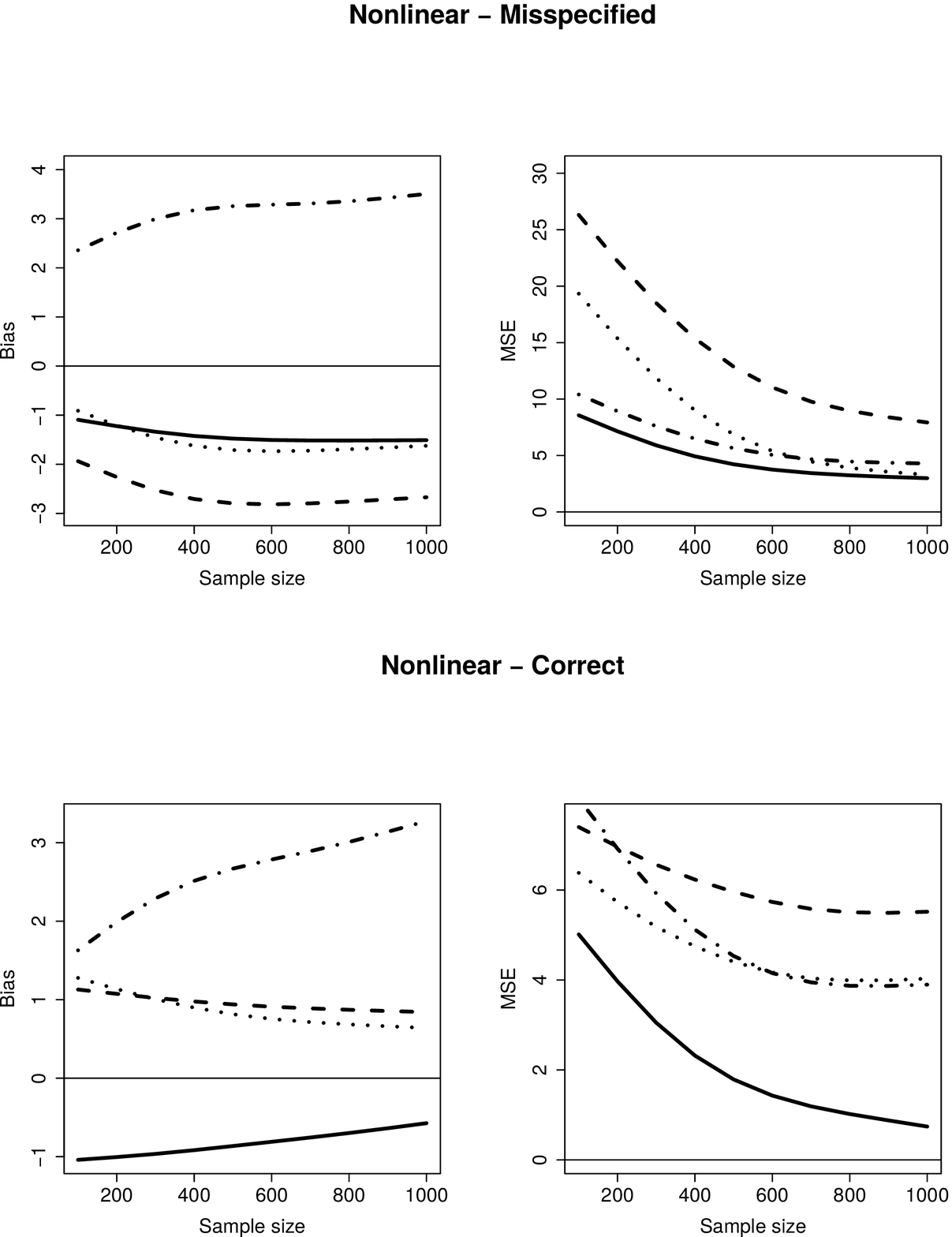}
\end{center}
\caption{\footnotesize Bias and MSE of the estimated time-varying treatment effect using KOW (solid), IPTW (dashed), sIPTW (dotted) and CBPS (dashed-dotted) when increasing the sample size from $n=100$ to $n=1{,}000$, in the nonlinear-misspecified scenario (top panels) and in the nonlinear-correct scenario (bottom panels).
\label{fig2} }
\end{figure}

\subsubsection{Across values of the penalization parameter, $\lambda$}

Figures \ref{fig1b} and \ref{fig2b} show the ratios of squared biases (left panels) and of MSEs (right panels) when comparing KOW with IPTW (solid), sIPTW (dashed) and CBPS (dotted) across different values of $\lambda$
and $n=500$ in the linear and nonlinear scenarios, respectively.
Values above 1 means that KOW had a lower bias or MSE. 
For zero or small $\lambda$, KOW significantly outperformed IPTW, sIPTW and CBPS with respect to bias. In many cases, the MSE was also smaller for zero $\lambda$. But, the biggest benefit in MSE was seen for larger $\lambda$.
The peaks of the right panels represent the points for which $\lambda$ is optimal, \textit{i.e.}, the MSE of KOW is minimized. The solid vertical lines on the right panels show the mean values across replications of the $\lambda$ value obtained by the procedure described in Section \ref{guidelines} and \ref{simu_setup} as done in the previous section. It can be seen that these are very near the points at which the MSE is minimized. The benefit in MSE both at and around this point was significant across all scenarios.

\begin{figure}[H]
\begin{center}
\includegraphics[scale=.6]{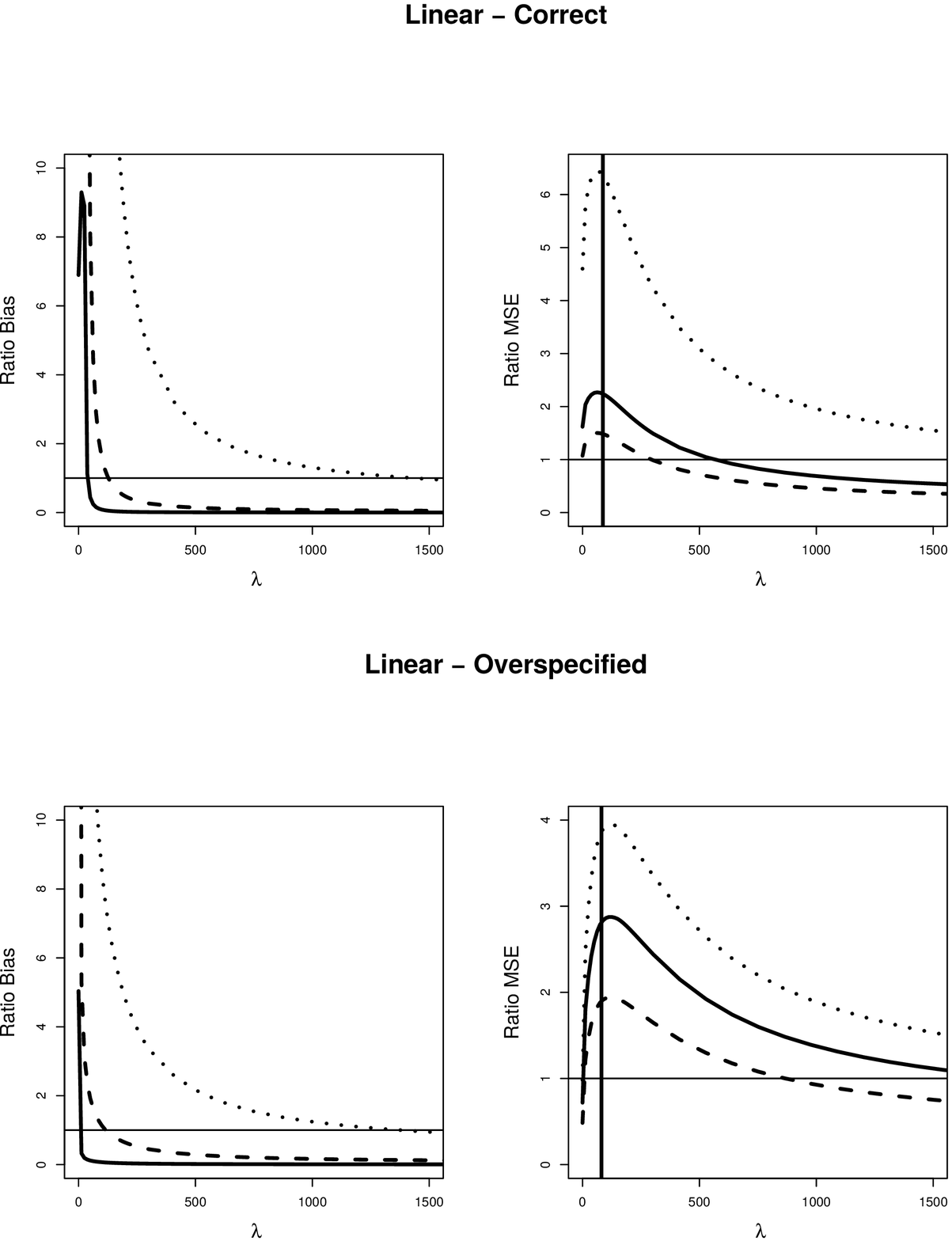}
\end{center}
\caption{\footnotesize Ratios of squared biases and MSEs comparing KOW with IPTW (solid), sIPTW (dashed) and CBPS (dotted) across values of $\lambda=0, \ldots, 1500$ in the linear-correct scenario (top panels) and  in the linear-overspecified scenario (bottom panels). Ratios above 1 means that KOW had a lower bias or MSE. Vertical bars show the mean value of $\lambda$, across simulations, obtained as described in Section \ref{simu_ss}. 
\label{fig1b} }
\end{figure}

\begin{figure}[H] 
\begin{center}
\includegraphics[scale=.6]{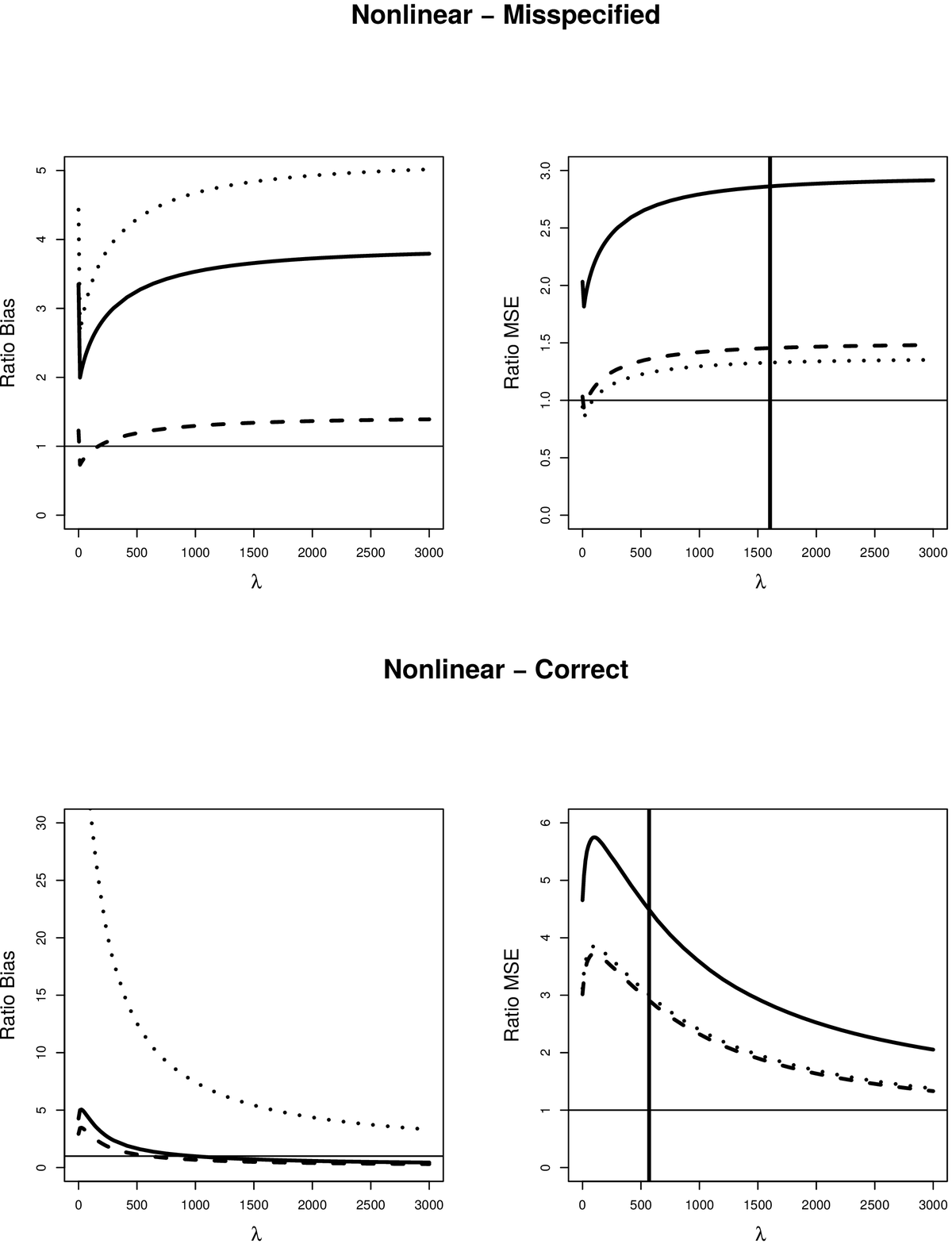}
\end{center}
\caption{\footnotesize Ratios of squared biases and MSEs comparing KOW with IPTW (solid), sIPTW (dashed) and CBPS (dotted) across values of $\lambda=0, \ldots, 3{,}000$ in the nonlinear-misspecified scenario (top panels) and in the nonlinear-correct scenario (bottom panels). Ratios above 1 means that KOW had a lower bias or MSE. Vertical bars show the mean value of $\lambda_=0, \ldots, 3{,}000$, across simulations, obtained as described in Section \ref{simu_ss}.
\label{fig2b} }
\end{figure}

\subsubsection{Computational time of KOW}
\label{simu_comp_time}

In this section we present the results of a simulation study aimed at comparing the mean computational time of KOW and CBPS. Compared to sIPTW based on pooled logistic regression, which is generally very fast, both KOW and CBPS have a nontrivial computational time that can grow with both the total number time periods $T$ and the number of covariates (which, for KOW, manifests as the complexity of the kernel functions). For KOW, the most time-consuming tasks are tuning $\lambda$ by marginal likelihood and computing the matrices that define problem \eqref{QP}, which are affected by these two factors, while solving problem \eqref{QP} is fast and does not depend on those factors. CBPS computational time is dominated by inverting a covariance matrix with dimensions increasing exponentially in $T$ and linearly in the number of covariates. \cite{imai2015robust} also propose using an approximate low-rank matrix that ignores certain covariance terms to make the matrix inversion faster.

Here we compare KOW, CBPS with full covariance matrix (CBPS-full), and CBPS with its low-rank approximation (CBPS-approx) when increasing the number of time periods and the number of covariates. Specifically, following the linear-correct scenario presented in Section \ref{simu_setup}, we fixed the sample size equal to $n=100$ and randomly generated 100 samples considering $T=3, \ldots, 10$, and $p=3,\ldots,8$, where $p$ is the total number of covariates $X_{t}$ for each $t$. We fixed the number of covariates to be equal to $p=3$ when evaluating the mean computational times over time periods, while we fixed the number of time periods to be equal to $T=5$ when analyzing over the number of covariates.
For each sample, we computed the KOW weights by solving eq.~\eqref{QP} using kernel $\mathcal K_1$. We used Gaussian process marginal likelihood to tune the kernels' hyperparameters and penalization parameter. We computed CBPS weights using the linear CBPS as in Section \ref{simu_setup}. We used the \textsf{R} package \textsf{rbenchmark} to compute the computational time on a PC with an i7-3770 processor, 3.4 GHz, 8GB RAM and a Linux Ubuntu 16.04 operating system. 

Solid lines of Figure  \ref{fig3} represent mean computational times for KOW, dashed for CBPS-full, and dotted for CBPS-approx.  When the number of time periods was relatively small, the mean computational time of KOW was higher compared with both CBPS methods (left panel of Figure \ref{fig3}). However, the mean computation time of KOW over time periods increased linearly while that of both CBPS methods increased exponentially. This is due to the fact that, as presented in Section \ref{imbalance}, the number of imbalances that we need to  minimize grows linearly in the number of time periods. The mean computational time required by KOW when increasing the number of covariates remained constant, while it increased for both CBPS-full and CBPS-approx, with CBPS-full increasing more rapidly. In summary, KOW was less affected by the total number of time periods and covariates compared with CBPS with full and low-rank approximation matrix. 

Computing KOW required three steps: tuning the parameters, constructing the matrices for problem \eqref{QP}, and solving problem \eqref{QP}. On average, for $T=3$, the first step required 21\% of the total computational time, the second 78.8\%, and the last 0.2\%.
Thus, solving the optimization problem itself is very fast and is not the bottleneck. 



\begin{figure}[h!] 
\begin{center}
\includegraphics[trim={0 13cm 0 0},clip,scale=.6]{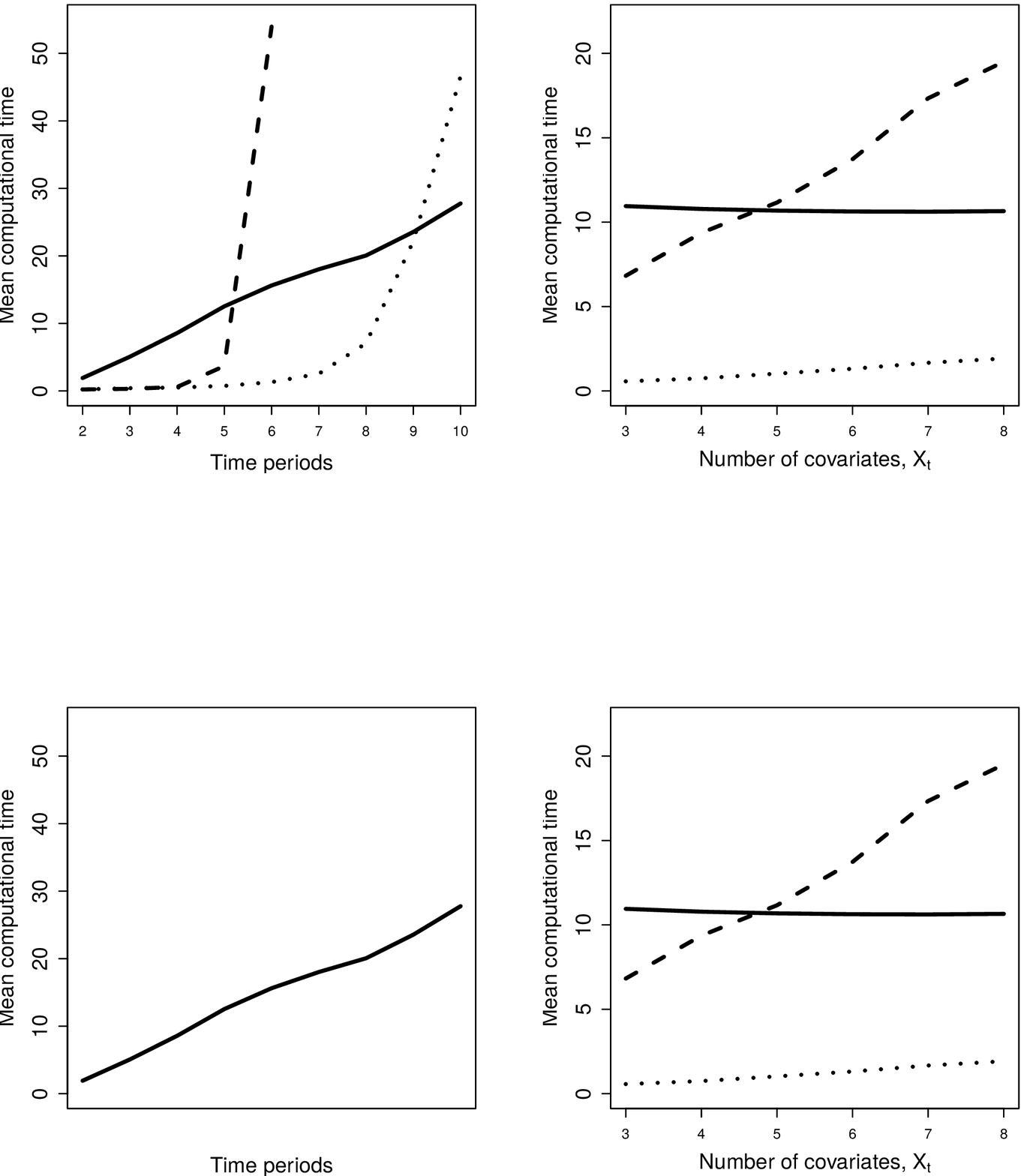}
\end{center}
\caption{{\footnotesize Mean computational time in seconds of KOW (solid), CBPS with full covariate matrix (dashed), and CBPS with the low-rank approximation of the full matrix (dotted) over time periods when $n=100$, $p=3$ and $T=2,\ldots,10$ (left panel) and over number of covariates, when $n=100$, $T=5$ and $p=3,\ldots,8$ (right panel).}
\label{fig3} }
\end{figure}

\section{
{KOW with informative censoring}}
\label{censor}



In longitudinal studies, participants may drop out the study before the end of the follow-up time and their outcomes are, naturally, missing observations. When this missingness is due to reasons related to the study (\textit{i.e.}, related to the potential outcomes), selection bias is introduced. This phenomenon is referred to as informative censoring and it is common in the context of survival analysis where the interest is on analyzing time-to-event outcomes. Under the assumptions of consistency, positivity, and sequential ignorability of both treatment and censoring, \cite{robins1999estimation} showed that a consistent estimate of the causal effect of a time-varying treatment can be obtained by weighting each subject 
$i=1, \ldots, n$ 
at each time period 
by the product of inverse probability of treatment and censoring weights. Inverse probability of treatment weights are obtained as presented in Section \ref{sec:rew_msm}, while inverse probability of censoring weights are usually obtained by inverting the probability of being uncensored at time $t$, given the treatment and confounder history up to time $t$ \citep{hernan2001marginal}.  


In this section we extend KOW to similarly handle informative censoring.
We demonstrate that under sequentially ignorable censoring, minimizing the very same discrepancies as before at each time period, restricted to the units for which data is available, actually controls for both time-dependent confounding as well as informative censoring. Thus, KOW naturally extends to the setting with informative censoring.

Let $C_{it}\in\{0,1\}$ for $t=1,\dots,T$ indicate whether unit $i$ is censored in time period $t$ and let $C_{i0}=0$. 
Note that $C_{it}=1$ implies that $C_{i,t+1}=1$ and that $C_{it}=0$ implies that $C_{i,t-1}=0$. 
All we require is that we (at least) 
observe outcomes $Y_i$ whenever $C_{iT}=0$,
$X_{it}$ whenever $C_{i,t-1}=0$,
and $A_{it}$ whenever $C_{it}=0$.
Note we might observe more, such as the treatment at time $t$ 
for a unit with $C_{i,t-1}=0$, or perhaps only some of the data after censoring is corrupted, but that is not required. 
We summarize the assumption of sequentially ignorable censoring as
\begin{equation}\label{ignocens}
Y(\overline a)\independent \overline C_{t} \mid \overline A_{t},\overline X_{t}.
\end{equation}

Let us redefine
\begin{align}
\label{deltasc}
\delta^{(1)}_{a_1}(W,h^{(1)}) &= \mathbbm{E} \left[ W \mathbbm{1}[A_1=a_1]\I[C_1=0] h^{(1)}(X_1) \right] - \mathbbm{E} \left[ h^{(1)}(X_1) \right] \\\notag
g_{\overline a}^{(1)}(X_1) &=  \mathbbm{E} \left[ Y(\overline a) \mid X_1 \right],\\\notag
\delta^{(t)}_{a_t}(W,h^{(t)}) &= \mathbbm{E} \left[ W \mathbbm{1}[A_t=a_t]\I[C_t=0] h^{(t)}(\overline A_{t-1},\overline X_t) \right] \\ \notag &\phantom{=}- \mathbbm{E} \left[ W\I[C_{t-1}=0] h^{(t)}(\overline A_{t-1},\overline X_t) \right],&&\forall t\geq2, \\\notag
g_{\overline a}^{(t)}(\overline A_{t-1},\overline X_t) &=  \mathbbm{1}[\overline A_{t-1}=\overline a_{t-1}]\mathbbm{E} \left[ Y(\overline a) \mid \overline A_{t-1},\overline X_t \right],&&\forall t\geq2.
\end{align}

Similarly to Theorem \ref{thm1}, the following theorem shows that we can write the difference between the weighted average outcome
among the  \emph{uncensored} $\overline a$-treated units, $\mathbbm{E} \left[ W \mathbbm{1}[\overline A=\overline a]\mathbbm{1}[C_T=0] Y \right]$, and the true average potential outcome of $\overline a$, $\mathbbm{E} \left[ Y(\overline a) \right]$,  as the sum over time points $t$ of discrepancies involving the values of treatment and confounder histories up to time $t$.
\begin{theorem}
\label{thm3}
Under 
assumptions \eqref{positivity}--\eqref{igno} and \eqref{ignocens}, 
\begin{equation}\label{decomp_cens}
\mathbbm{E} \left[ W \mathbbm{1}[\overline A=\overline a]\mathbbm{1}[C_T=0] Y \right]
- \mathbbm{E} \left[ Y(\overline a) \right]
=
\sum_{t=1}^T\delta^{(t)}_{a_t}(W,g_{\overline a}^{(t)}).
\end{equation}
\end{theorem}



We then define 
the empirical counterparts to $\delta^{(t)}_{a_t}(W,h^{(t)})$ as before in eq.~\eqref{deltase} but limit ourselves to \textit{uncensored} units, as in eq.~\eqref{deltasc}.
We similarly define imbalance,
$\text{IMB}(W_{1:n};(\overline g_{\overline a}^{(t)})_{\overline a\in \mathcal A})$,
and the worst case imbalance $\mathcal{B}^2(W_{1:n})$,
as before in eqs.~\eqref{biasb2} and \eqref{nswcimb}.
Finally, again using kernels to specify norms, we obtain weights that optimally balance covariates to control for time-dependent confounding and account for informative censoring while controlling precision by solving the quadratic optimization problem, 
\begin{equation}
\label{QP_C}
	\begin{aligned}
\underset{W_{1:n} \in \mathcal{W}}{\min} \quad & 
\frac12W_{1:n}^TK_{\lambda}^\circ W_{1:n}-e^TK_{\lambda}W_{1:n}
,
	\end{aligned}
\end{equation}
where 
$K_{\lambda}^\circ = K^\circ+2\lambda I$, 
$K_{\lambda}=K_1+2\lambda I$, 
$K^\circ = \sum_{t=1}^TK_t^\circ$,  
\break $K_t^\circ=\sum_{a_t \in \lbrace 0,1 \rbrace} I^{(t)}_{a_t}K_tI^{(t)}_{a_t}$, 
$K_t\in\mathbb R^{n\times n}$ is a semidefinite positive matrix defined as $K_{tij}=\mathcal K_t((\overline A_{i,t-1},\overline X_{it}),(\overline A_{j,t-1},\overline X_{jt}))$, 
$I^{(t)}_{a_t}$ is the diagonal matrix with $\mathbb I[A_{it}=a_t]\mathbb I[C_{it}=0]-\mathbb I[C_{i,t-1}=0]$ in its $i^\text{th}$ diagonal entry (recall $C_{i,0}=0$ for all $i$), 
and $e$ is the vector of all ones. 


\section{Applications}\label{empirics}

In this section, we present two empirical applications of KOW. In the first, we estimate the effect of treatment initiation on time to death among people living with HIV (PLWH). In the second, we evaluate the impact of negative advertising on election outcomes.

\subsection{The effect of HIV treatment on time to death}
\label{casehiv}

In this section, we analyze data from the Multicenter AIDS Cohort Study (MACS) to study the effect of the initiation time of treatment on time to death among PLWH. Indeed, due to the longitudinal nature of HIV treatment and the presence of time-dependent confounding, MSMs have been widely used to study causal effects in this domain
\citep[among others]{hernan2000marginal, hernan2001marginal, hiv2010effect, hiv2011initiate, lodi2017effect}.
As an example of time-dependent confounding, CD4 cell count, a measurement used to monitor immune defenses in PLWH and to make clinical decisions, is a predictor of both treatment initiation and survival, as well as being itself influenced by prior treatments. 
Recognizing the censoring in the MACS data,
\cite{hernan2000marginal} showed how to estimate the parameters of the MSM by inverse probability of treatment and censoring weighting (IPTCW). 

Here, we apply KOW as proposed in Section \ref{censor} to handle both time-dependent confounding and informative censoring while controlling precision. We considered the following potential time-dependent confounders associated with the effect of treatment initiation and the risk of death: CD4 cell count, white blood cell count, red blood cell count, and platelets. We also identified the age at baseline as a potential time-invariant confounding factor.  
We considered only recently developed HIV treatments, thus, including in the analysis only PLWH that started treatment after 2001. The final sample was comprised of a total of $n=344$ people and 760 visits, with a maximum of $T=8$ visits per person. We considered two sets of KOW weights, either obtained by using a product of (1) two linear kernels, one for the treatment history and one for the confounder history ($\mathcal{K}_1$) or (2) a linear kernel for the treatment history and a polynomial kernel of degree 2 for the confounder history ($\mathcal{K}_2$). We scaled the covariates related to the treatment and confounder history, and tuned the kernels' hyperparameters and the penalization parameter by using Gaussian processes marginal likelihood as presented in Section \ref{guidelines}.   Following previous approaches studying the HIV treatment using IPTCW that modeled treatment and censoring using single time lags \citep{hernan2000marginal,hernan2001marginal,hernan2002estimating}, we included in each kernel the time-invariant confounders, the previous treatment, $A_{t-1}$, and the time-dependent confounders at time $t$, $X_t$, instead of the entire histories. 
As described in Section \ref{guidelines}, since we have repeated observations of outcomes, we compute a set of KOW weights by solving the optimization problem (\ref{QP_C}) for each horizon up to $T$.
In addition, as described in Section \ref{guidelines}, we constrained the mean of the weights to be equal to one. 

We compared the results obtained by KOW with those from IPTCW and stabilized-IPTCW (sIPTCW). 
The latter sets of weights 
were obtained by using a logistic regression on the treatment history and the aforementioned time-invariant and time-dependent confounders and using only one time lag for each of the treatment and time-dependent confounders
as done in previous approaches studying the HIV treatment using IPTCW \citep{hernan2000marginal,hernan2001marginal,hernan2002estimating}. 
The numerator of sIPTCW was computed 
by modeling $h(\overline{A}_t)$ in eq.~(\ref{sipweights}) with a logistic regression on the treatment history only using one time lag. 
We modeled the inverse probability of censoring weights similarly. 
The final sets of IPTCW and sIPTCW weights were obtained by multiplying inverse probability of treatment and censoring weights. 
%
We did not compare the results with those of CBPS
because it does not handle informative censoring. In particular, CBPS requires a complete $n \times T$ matrix of observed time-dependent confounders, while in the MACS dataset many entries are missing.  

We estimated the hazard ratio of the risk of death by using a weighted Cox regression model \citep{hernan2000marginal} weighted by KOW, IPTCW, or sIPTCW and using robust standard errors \citep{freedman2006so}. We used \textsf{Gurobi} and its \textsf{R} interface to solve eq.~\eqref{QP_C} and obtain the KOW weights, the \textsf{Matlab} package \textsf{GPML} to perform the marginal likelihood estimation of hyperparameters, {the \textsf{R} package \textsf{R.matlab} to call \textsf{MatLab} from within \textsf{R},} the \textsf{R} package \textsf{ipw} \citep{van2011ipw} to fit the treatment models for IPTCW and sIPTCW, and the \textsf{R} command \textsf{coxph} (with robust variance estimation) to fit the outcome model. It took 13.5 seconds to obtain a solution for KOW.  Table \ref{table_hiv} summarizes the result of our analysis. Both KOW ($\mathcal{K}_1$) and ($\mathcal{K}_2$) showed a significant protective effect of HIV treatment on time to death among PLWH. IPTCW showed a similar effect but with lower precision, resulting in a non-significant effect. With similar precision obtained by KOW, sIPTCW showed a non-significant effect of HIV treatment on time to death.  
{Whereas analyses based on IPTCW and sIPTCW lead to non-significant and inconsistent conclusions, the results we obtained by using KOW show that  PLWH can benefit from HIV treatment, as shown in independent
randomized placebo-controlled trials \citep{cameron1998randomised,hammer1997controlled}. }


\begin{table}[H]
\centering
\caption{Effect of HIV treatment on time to death.}
\label{table_hiv}
\begin{threeparttable}
\begin{tabular}{ccccc}
\hline

 &         \multicolumn{2}{c}{KOW}                         & \multicolumn{2}{c}{Logistic}                                                        \\
& \multicolumn{1}{c}{$\mathcal{K}_1$} & \multicolumn{1}{c}{$\mathcal{K}_2$} & \multicolumn{1}{c}{IPTCW} & \multicolumn{1}{c}{sIPTCW} \\

\hline
\textit{$\hat{HR}$ } & 0.40* & 0.48*  & 0.14 & 1.25   \\ 
\textit{ SE } & (0.30) & (0.28) & (1.15)  & (0.30) 
\end{tabular}
\begin{tablenotes}
      \footnotesize
      \item Note: $\hat{HR}$ is the estimated hazard ratio of the effect of HIV treatment initiation on time to death. $SE$ is the estimated robust standard error. Weights were obtained by using, KOW  ($\mathcal{K}_1$): a product of two linear kernels, one for the treatment history and one for the confounder history; KOW ($\mathcal{K}_2$): a product between a linear kernel for the treatment history and a polynomial kernel of degree 2 for the confounder history; IPTCW:  a logistic regression on the treatment history and the time-invariant and time-dependent confounders (using only one time lag for each of the treatment and time-dependent confounders); sIPTCW: stabilized IPTCW.  * indicates statistical significance at the 0.05 level.
    \end{tablenotes}
\end{threeparttable}
\end{table}


\subsection{The impact of negative advertising on election outcomes}
\label{caseblack}

In this section, we analyze a subset of the dataset from \cite{blackwell2013framework} to estimate the impact of negative advertising on election outcomes.
Because of the dynamic and longitudinal nature of the problem and presence of time-dependent confounders, MSMs have been used previously used to study the question 
\citep{blackwell2013framework}. 
Specifically, poll numbers are time-dependent confounders as they might both be affected by negative advertising and might also affect future poll numbers. 
We constructed the subset
of the data from \citet{blackwell2013framework} by considering
the five weeks leading up to each of 114 elections held 2000--2006 (58 US Senate, 56 US gubernatorial). 
Differently from Section \ref{casehiv} in which the outcome was observed at each time period, in this analysis, the binary election outcome was observed only at the end of each five-week trajectory. 
In addition, all units were uncensored.

We estimated the parameters of two MSMs, the first having separate coefficients for negative advertising in each time period and the second having one coefficient for the cumulative effect of negative advertising. 
Each MSM was fit using weights given by each of KOW, IPTW, sIPTW, and CBPS (both full and approximate).
We used the following time-dependent confounders: Democratic share of the polls, proportion of undecided voters, and campaign length. We also used the following time-invariant confounders: baseline Democratic vote share, proportion of undecided voters, status of incumbency, election year and type of office. 
We obtained two sets of KOW weights by using a product of (1) two linear kernels, one for the history of negative advertising and one for the confounder history ($\mathcal{K}_1$) and (2) a linear kernel for the history of negative advertising and a polynomial kernel of degree 2 for the confounder history ($\mathcal{K}_2$).
The kernels were over the complete confounder history up to time $t$, 
$\overline X_t$, and two time-lags of treatment history, $A_{t-1},A_{t-2}$.
We scaled the covariates and tuned the kernels' hyperparameters and the penalization parameter by using Gaussian processes marginal likelihood. We obtained the final set of KOW weights by solving eq.~(\ref{QP}). 
We compared the results obtained by KOW with those from IPTW, sIPTW, CBPS-full, and CBPS-approx. To obtain the sets of IPTW, sIPTW, and CBPS weights, we used logistic models conditioned on the confounder history and two time-lags from the treatment history. To compute the numerator of sIPTW weights, we used a logistic regression conditioned only on two time-lags from the treatment history. 
We used \textsf{Gurobi} and its \textsf{R} interface to solve eq.~\eqref{QP_C} and obtain the KOW weights, the \textsf{Matlab} package \textsf{GPML} to perform the marginal likelihood estimation of hyperparameters, {the \textsf{R} package \textsf{R.matlab} to call \textsf{MatLab} from within \textsf{R},} the \textsf{R} command \textsf{glm} to fit the treatment models for IPTW and sIPTW, the \textsf{R} package \textsf{CBMSM} for CBPS, the \textsf{R} command \textsf{lm}  to fit the outcome model, and the \textsf{R} package \textsf{sandwich} to estimate robust standard errors. The computational time to obtain a solution was equal to 12.6 seconds for KOW, while it was equal to 104 seconds for CBPS-full and 3.8 seconds for CBPS-approx.

Table \ref{table_black} summarizes the results of our analysis, reporting robust standard errors \citep{freedman2006so}. 
The first six rows of Table \ref{table_black} show the effect of the time-specific negative advertising. The last two rows present the effect of the cumulative effect of negative advertising. KOW ($\mathcal{K}_1$ and $\mathcal{K}_2$) and IPTW showed similar effects, with increased precision when using KOW except for time 4, in which both methods showed a significant negative effect but with greater precision when using IPTW.  sIPTW, CBPS-full and CBPS-approx showed a significant negative effect at time 3 with similar precision. No significant results were obtained when considering the cumulative effect of negative advertising. All except sIPTW, showed a negative cumulative effect. KOW ($\mathcal{K}_1$) was the most precise. 
{We conclude that, the impact of negative advertising in the majority of the time periods and its cumulative effect on election outcomes are not statistically significant.}

\begin{table}[H]
\centering
\caption{Impact of negative advertising on election outcomes.}
\label{table_black}
\begin{threeparttable}
\begin{tabular}{ccccccc}
\hline
 \multicolumn{1}{c}{$\hat \beta$}     &         \multicolumn{2}{c}{KOW}                         & \multicolumn{2}{c}{Logistic}                            & \multicolumn{2}{c}{CBPS}                                  \\
\multicolumn{1}{c}{$SE$} & \multicolumn{1}{c}{$\mathcal{K}_1$} & \multicolumn{1}{c}{$\mathcal{K}_2$} & \multicolumn{1}{c}{IPTW} & \multicolumn{1}{c}{sIPTW} & \multicolumn{1}{c}{Full} & \multicolumn{1}{c}{Approx} \\
\hline
Intercept   			& 54.54*      	&  53.84*    	& 53.05*  	& 47.46*  	&   51.25* 	&   52.17* \\
							& (2.15) 		& (2.38) 		&  (2.88) & (2.98) 	& (2.70) 				& (2.39)  \\
Negative$_1$ 		& 2.43        	&   3.27  		& 4.41  	& 7.62*  	&   5.95*	 	&  4.81* \\
							& (1.86) 		& (1.86)		&  (2.56) & (3.26) 	& (2.49) 				& (2.22) \\
Negative$_2$ 		& 3.73     	&  3.24       & 5.51*  	& 3.17  	&   3.55 	&   2.65 \\
							& (2.18) 		& (2.22)		&  (2.38) & (3.19) 	& (2.73) 				& (2.42) \\
Negative$_3$ 		& -2.51    	&  -2.39       & -4.37  	& -8.32*  	&   -6.50* 	&   -6.31* \\
							& (2.34) 		& (2.45)		&  (2.54) & (3.84) 	& (3.20) 				& (3.24) \\
Negative$_4$ 		& -7.16*      &   -7.22*    & -8.77*  & 2.34  	&   -3.55 	&   -3.12 \\
							& (2.57) 		& (2.75) 		&  (1.54) & (3.11) 	& (3.71) 				& (3.59) \\
Negative$_5$ 		& -2.75*      & -1.79       	& -3.19  	& -3.62  	&   -1.92 	& -1.96   \\
							& (1.42) 		& (1.59) 		&  (2.19) & (2.59) 	& (1.62) 				& (1.54) \\
							& 		& 	&   & 	& 				&  \\
\hline
Intercept   			& 51.40*     &  50.56*       & 58.29*  & 42.63*  	&   49.38* 	& 50.28*   \\
							& (2.45) 		& (2.63)  	&  (4.29) & (4.15) 			& (2.68) 				& (2.49)\\
Cumulative  			& -0.59     	&  -0.37       & -0.93  	& 1.91  	&   -0.28 	&   -0.41 \\
							& (0.58)  	& (0.64)  	&  (1.57) & (1.15) 			& (0.65) 				& (0.77) \\
\hline
\end{tabular}
\begin{tablenotes}
      \footnotesize
      \item Note: $\hat{\beta}$ is the estimated effect of negative advertising. $SE$ is the estimated robust standard error. Weights were obtained by using, KOW  ($\mathcal{K}_1$): a product of two linear kernels, one for the history of negative advertising and one for the confounder history; KOW ($\mathcal{K}_2$): a product between a linear kernel for the history of negative advertising and a polynomial kernel of degree 2 for the confounder history; IPTW: a logistic model conditioned on the confounder history and two time-lags from the treatment history; sIPTW: stabilized IPTW; CBPS-full: CBPS with full covariance matrix; CBPS-approx: CBPS with low-rank approximation. * indicates statistical significance at the 0.05 level.
    \end{tablenotes}
    \end{threeparttable}
\end{table}

\section{Conclusions}\label{conclusions}


In this paper we presented KOW, which optimally finds weights for fitting an MSM with the aim of balancing time-dependent confounders while controlling for precision. That KOW uses mathematical optimization to directly and fully balance covariates 
as well as optimize precision
explains the better performance of KOW over IPTW, sIPTW and CBPS observed in our simulation study. 
{In addition, as shown in Sections \ref{nswci}, \ref{simu} and \ref{censor}, the proposed methodology only needs to minimize a number of discrepancies that grows linearly in the number of time periods, mitigates the possible misspecification of the treatment assignment model, allows balancing non-additive covariate relationships, and can be extended to control for informative censoring, which is a common feature of longitudinal studies.}  

Alternative formulations of our imbalance-precision optimization problem, eq.~\eqref{cmse}, may be investigated. For example, additional linear constraints can be added to the optimization problem, as shown in the empirical application of Section \ref{casehiv}, and different penalties can be considered to control for extreme weights. For instance, in eq.~\eqref{cmse}, at the cost of no longer being able to use convex-quadratic optimization, one may directly penalize the covariance matrix of the weighted least-square estimator rather than use a convex-quadratic surrogate as we do. 

One may also change the nature of precision control. Here, we suggested penalization in an attempt to target total error.
Alternatively, similar to \cite{santacatterina2017optimal}, we may reformulate eq.~\eqref{cmse} as a constrained optimization problem where the precision of the resulting estimator is constrained by an upper bound $\xi$, thus seeking to minimize imbalances subject to having a bounded precision. In our convex formulation, the two are equivalent by Lagrangian duality in that for every precision penalization $\lambda$ there is an equivalent precision bound $\xi$. However, it may make specifying the parameters easier depending on the application as it may be easier for a practitioner to conceive of a desirable bound on precision. There may also be other ways to choose the penalization parameter. Here we suggested using maximum marginal likelihood but cross validation based on predicting outcomes and their partial means may also be possible. 

The flexibility of our approach is that any of these changes amount to simply modifying the optimization problem that is fed to an off-the-shelf solver. Indeed, we were able to extend KOW from the standard longitudinal setting to also handle both repeated observations of outcomes and informative censoring. In addition to offering flexibility, the optimization approach we took, which directly and fully minimized our error objective phrased in terms of covariate imbalances, was able to offer improvements on the state of the art.

\bibliographystyle{chicago}
\bibliography{bib}


\newpage

  \section*{Appendix}

\begin{proof}[Proof of Theorem~\ref{thm1}]
\label{proof_deco}
For clarity, we prove this for $T=2$. The extension to $T>2$ is by induction.
Under consistency and assumptions \eqref{positivity}--\eqref{igno}, we have
\begin{align*}
&\E[W\I[\overline A=\overline a]Y]\\
&=\E[W\I[A_1=a_1]\I[A_2=a_2]Y(\overline a)]&&\text{(consistency)}\\
&=\E[W\I[A_1=a_1]\I[A_2=a_2]\E[Y(\overline a)\mid A_1,A_2,X_1,X_2]]&&\text{(iterated expectations)}\\
&=\E[W\I[A_1=a_1]\I[A_2=a_2]\E[Y(\overline a)\mid A_1,X_1,X_2]]&&\text{(sequential ignorability)}\\
&=\E[W\I[A_1=a_1]\E[Y(\overline a)\mid A_1,X_1,X_2]]+\delta^{(2)}_{\overline a}(W,g_{\overline a}^{(2)})&&\text{(definition of $\delta^{(2)}_{\overline a},g_{\overline a}^{(2)}$)}\\
&=\E[W\I[A_1=a_1]\E[Y(\overline a)\mid A_1,X_1]]+\delta^{(2)}_{\overline a}(W,g_{\overline a}^{(2)})&&\text{(iterated expectations)}\\
&=\E[W\I[A_1=a_1]\E[Y(\overline a)\mid X_1]]+\delta^{(2)}_{\overline a}(W,g_{\overline a}^{(2)})&&\text{(sequential ignorability)}\\
&=\E[\E[Y(\overline a)\mid X_1]]+\delta^{(1)}_{\overline a}(W,g_{\overline a}^{(1)})+\delta^{(2)}_{\overline a}(W,g_{\overline a}^{(2)})&&\text{(definition of $\delta^{(1)}_{\overline a},g_{\overline a}^{(1)}$)}\\
&=\E[Y(\overline a)]+\delta^{(1)}_{\overline a}(W,g_{\overline a}^{(1)})+\delta^{(2)}_{\overline a}(W,g_{\overline a}^{(2)})&&\text{(iterated expectations)}
\end{align*}

\end{proof}


\newpage

\begin{proof}[Proof of Thm. \ref{thm2}]
\label{proof_QP}
Define $K_{tij}=\mathcal K_t((\overline A_{i,t-1},\overline X_{it}),(\overline A_{j,t-1},\overline X_{jt}))$. Then, by the representer property of the kernels and by self-duality of Hilbert spaces,

\begin{equation}
	\begin{aligned}
			\Delta^{(1)}_{a_1}(W_{1:n})^2 &= \sup_{\|h^{(1)}\|_{(1)}^2\leq 1} \left( \hat\delta^{(1)}_{a_1}(W_{1:n},h^{(1)}) \right)^2			\\
			&= \text{sup}_{\| h^{(1)} \|^2_{(1)}  \leq 1}  \left(  \dfrac{1}{n} \sum_{i=1}^n \underbrace{\left(  W_i\mathbbm{1}[A_t=a_t] - 1 \right)}_{z_i}h^{(1)}(X_{i1}) \right)^2 \\
 			&= \text{sup}_{ \|h^{(1)} \|^2_{(1)}  \leq 1} \left( \dfrac{1}{n}\sum_{i=1}^n z_i \langle \mathcal{K}_{t}((X_{i1}),\cdot), h^{(1)}(X_{i1})  \rangle \right)^2 \\
&=  \left\| \dfrac{1}{n}\sum_{i=1}^n z_i  \mathcal{K}_{1}((X_{i1}),\cdot) \right\|_{(1)}^2 \\ 
&=  \left\langle \dfrac{1}{n}\sum_{i=1}^n z_i  \mathcal{K}_{t}((X_{i1}),\cdot),  \dfrac{1}{n}\sum_{i=1}^n z_i  \mathcal{K}_{t}((X_{i1}),\cdot) \right\rangle \\
&= \dfrac{1}{n^2}\sum_{i=1}^n\sum_{j=1}^n z_i z_j  \mathcal{K}_{t}((X_{i1}),(X_{j1})) \\\notag
&=\frac1{n^2}\sum_{i=1}^n\sum_{j=1}^n(W_i\mathbb I[A_{i1}=a_1]-1)(W_j\mathbb I[A_{j1}=a_1]-1)K_{tij}
\\\notag
&=\frac1{n^2}(I^{(1)}_{a_1}W_{1:n}-e)^TK_t(I^{(1)}_{a_1}W_{1:n}-e)\\\notag \\\notag
&=\frac1{n^2}W_{1:n}^TI^{(1)}_{a_1}K_1I^{(1)}_{a_1}W_{1:n}-2e^TK_1I^{(1)}_{a_1}W_{1:n}+e^TK_1e
  	\end{aligned}
\end{equation}

\begin{equation}
	\begin{aligned}
			\Delta^{(t)}_{a_t}(W_{1:n})^2 &= \sup_{\|h^{(t)}\|_{(t)}^2\leq 1} \left( \hat\delta^{(t)}_{a_t}(W_{1:n},h^{(t)}) \right)^2			\\
			&= \text{sup}_{\| h^{(t)} \|^2_{(t)}  \leq 1}  \left(  \dfrac{1}{n} \sum_{i=1}^n \underbrace{\left(  \mathbbm{1}[A_t=a_t] - 1 \right)W_i}_{z_i}h^{(t)}(\overline{A}_{i,t-1},\overline{X}_{it}) \right)^2 \\
 			&= \text{sup}_{ \|h^{(t)} \|^2_{(t)}  \leq 1} \left( \dfrac{1}{n}\sum_{i=1}^n z_i \langle \mathcal{K}_{t}((\overline{A}_{i,t-1},\overline{X}_{it}),\cdot), h^{(t)}(\overline{A}_{i,t-1},\overline{X}_{it})  \rangle \right)^2 \\
&=  \left\| \dfrac{1}{n}\sum_{i=1}^n z_i  \mathcal{K}_{t}((\overline{A}_{i,t-1},\overline{X}_{it}),\cdot) \right\|_{(t)}^2 \\ 
&=  \left\langle \dfrac{1}{n}\sum_{i=1}^n z_i  \mathcal{K}_{t}((\overline{A}_{i,t-1},\overline{X}_{it}),\cdot),  \dfrac{1}{n}\sum_{i=1}^n z_i  \mathcal{K}_{t}((\overline{A}_{i,t-1},\overline{X}_{it}),\cdot) \right\rangle \\
&= \dfrac{1}{n^2}\sum_{i=1}^n\sum_{j=1}^n z_i z_j  \mathcal{K}_{t}((\overline{A}_{i,t-1},\overline{X}_{it}),(\overline{A}_{j,t-1},\overline{X}_{jt})) \\\notag
&=\frac1{n^2}\sum_{i=1}^n\sum_{j=1}^n(W_i\mathbb I[A_{it}=a_t]-W_i)(W_j\mathbb I[A_{jt}=a_t]-W_j)K_{tij}
\\\notag
&=\frac1{n^2}(I^{(t)}_{a_t}W_{1:n}-W_{1:n})^TK_t(I^{(t)}_{a_t}W_{1:n}-W_{1:n})\\\notag
&=\frac1{n^2}W_{1:n}^T(I-I^{(t)}_{a_t}) K_t (I-I^{(t)}_{a_t}) W_{1:n}
  	\end{aligned}
\end{equation}

\end{proof}
 

\newpage

\begin{proof}[Proof of Theorem~\ref{thm3}]

For clarity, we prove this for $T=2$. The extension to $T>2$ is by induction.
Under consistency, assumptions \eqref{positivity}--\eqref{igno}, and assumption \eqref{ignocens},
\begin{align*}
&\E[W\I[\overline A=\overline a]\I[C_2=0]Y]\\
&=\E[W\I[A_1=a_1]\I[A_2=a_2]\I[C_2=0]Y(\overline a)]&&\text{(consistency)}\\
&=\E[W\I[A_1=a_1]\I[A_2=a_2]\I[C_2=0]\E[Y(\overline a)\mid A_1,A_2,X_1,X_2,C_2,C_1]]&&\text{(iterated expectations)}\\
&=\E[W\I[A_1=a_1]\I[A_2=a_2]\I[C_2=0]\E[Y(\overline a)\mid A_1,A_2,X_1,X_2]]&&\text{(eq.~\eqref{ignocens})}\\
&=\E[W\I[A_1=a_1]\I[A_2=a_2]\I[C_2=0]\E[Y(\overline a)\mid A_1,X_1,X_2]]&&\text{(eq.~\eqref{igno})}\\
&=\E[W\I[A_1=a_1]\I[C_1=0]\E[Y(\overline a)\mid A_1,X_1,X_2]]+\delta^{(2)}_{\overline a}(W,g_{\overline a}^{(2)})&&\text{(definition of $\delta^{(2)}_{\overline a},g_{\overline a}^{(2)}$)}\\
&=\E[W\I[A_1=a_1]\I[C_1=0]\E[Y(\overline a)\mid A_1,A_2,X_1,X_2]]+\delta^{(2)}_{\overline a}(W,g_{\overline a}^{(2)})&&\text{(eq.~\eqref{igno})}\\
&=\E[W\I[A_1=a_1]\I[C_1=0]\E[Y(\overline a)\mid A_1,A_2,X_1,X_2,C_1,C_2]]+\delta^{(2)}_{\overline a}(W,g_{\overline a}^{(2)})&&\text{(eq.~\eqref{ignocens})}\\
&=\E[W\I[A_1=a_1]\I[C_1=0]\E[Y(\overline a)\mid A_1,X_1,C_1]]+\delta^{(2)}_{\overline a}(W,g_{\overline a}^{(2)})&&\text{(iterated expectations)}\\
&=\E[W\I[A_1=a_1]\I[C_1=0]\E[Y(\overline a)\mid A_1,X_1]]+\delta^{(2)}_{\overline a}(W,g_{\overline a}^{(2)})&&\text{(eq.~\eqref{ignocens})}\\
&=\E[W\I[A_1=a_1]\I[C_1=0]\E[Y(\overline a)\mid X_1]]+\delta^{(2)}_{\overline a}(W,g_{\overline a}^{(2)})&&\text{(eq.~\eqref{igno})}\\
&=\E[\I[C_0=0]\E[Y(\overline a)\mid X_1]]+\delta^{(1)}_{\overline a}(W,g_{\overline a}^{(1)})+\delta^{(2)}_{\overline a}(W,g_{\overline a}^{(2)})&&\text{(definition of $\delta^{(1)}_{\overline a},g_{\overline a}^{(1)}$)}\\
&=\E[Y(\overline a)]+\delta^{(1)}_{\overline a}(W,g_{\overline a}^{(1)})+\delta^{(2)}_{\overline a}(W,g_{\overline a}^{(2)})&&\text{(iterated expectations)}
\end{align*}

\end{proof}

\end{document}